\documentclass[10pt]{amsart}
\usepackage{float}
\usepackage{color,colordvi}
\usepackage{url}
\usepackage{enumerate}
\usepackage{tabularx}
\usepackage{xcolor}
\usepackage{graphicx}

\newtheorem{theorem}{Theorem}[section]
\newtheorem{lemma}[theorem]{Lemma}
\newtheorem{cor}[theorem]{Corollary}

\theoremstyle{definition}
\newtheorem{definition}[theorem]{Definition}
\newtheorem{example}[theorem]{Example}

\theoremstyle{remark}
\newtheorem{remark}[theorem]{Remark}

\numberwithin{equation}{section}

\title[Inverse problems for locally perturbed  lattices]{Inverse problems for locally perturbed  lattices - Discrete Hamiltonian and
quantum graph}

\author{Emilia Bl{\aa}sten, Pavel Exner, Hiroshi Isozaki,\\ Matti Lassas and Jinpeng Lu}
\date{}

\AtEndDocument{\bigskip{\footnotesize%
  \textsc{Emilia Bl{\aa}sten: Department of Mathematics and Systems Analysis, Aalto University, FI-00076 Aalto, Finland} \par
  \textit{Email address}: \texttt{emilia.blasten@iki.fi} \par

  \addvspace{\medskipamount}
    \textsc{Pavel Exner: Doppler Institute for Mathematical Physics and Applied Mathematics, Czech Technical University in Prague, B\v{r}ehov\'a 7, 115 19 Prague, Czechia, and Nuclear Physics Institute, Czech Academy of Science, 250 68 \v{R}e\v{z}, Czechia} \par
    \textit{Email address}: \texttt{exner@ujf.cas.cz} \par

  \addvspace{\medskipamount}

  \textsc{Hiroshi Isozaki: Graduate School of Pure and Applied Sciences, Professor Emeritus, University of Tsukuba, Tsukuba, 305-8571, Japan} \par
  \textit{Email address}: \texttt{isozakih@math.tsukuba.ac.jp} \par

  \addvspace{\medskipamount}

  \textsc{Matti Lassas: Department of Mathematics and Statistics, University of Helsinki, FI-00014 Helsinki, Finland} \par
  \textit{Email address}: \texttt{matti.lassas@helsinki.fi} \par

  \addvspace{\medskipamount}
  \textsc{Jinpeng Lu: Department of Mathematics and Statistics, University of Helsinki, FI-00014 Helsinki, Finland} \par
  \textit{Email address}: \texttt{jinpeng.lu@helsinki.fi} \par
}}

\date{\today}

\begin{document}
\baselineskip 14pt
\maketitle

\begin{abstract}
We consider the inverse scattering problems for two types of Schr{\"o}dinger operators on  locally perturbed periodic lattices. For the discrete Hamiltonian,  the knowledge of the S-matrix for all energies determines the graph structure and the coefficients of the Hamiltonian. For locally perturbed equilateral metric graphs, the knowledge of the S-matrix for all energies determines the graph structure.
\end{abstract}

\section{Introduction}
\subsection{The goal of this work}

There are two basic models for describing the motion of quantum mechanical particles on a periodic lattice. In the first model, the configuration space consists of graph vertices only and the Hamiltonian is written as a difference operator which is determined by the adjacency matrix. We refer to this operator as the discrete Schr{\"o}dinger operator in this paper. In the other model, the wave functions are supported on the graph edges and the Hamiltonian is a differential operator on the edges. This model is called the quantum (or metric) graph.

The aim of this paper is twofold. The first topic concerns a locally perturbed periodic lattice. 
We analyze the discrete Schr{\"o}dinger operator having the form
\begin{equation}
\widehat H_G : \hat u \to \frac{1}{{\rm deg}\, v}\sum_{w\sim v, w \in G}g_{vw}\hat u(w) + q(v)\hat u(v), \quad v \in G,
\label{Equationinafiitepart}
\end{equation}
on a finite part of the graph $G$ and prove the following result (Theorem \ref{mainscattering}):
\begin{itemize}
\item Given a locally perturbed periodic lattice of a certain class 
and the associated discrete Hamiltonian $\widehat H_G$, we can determine the graph structure, $g_{vw}$ and $q(v)$ from the knowledge of the S-matrix for all energies.
\end{itemize}
Here, a local perturbation of lattice means replacing a finite number of edges and vertices by a finite number of other edges and vertices and changing the weights $g_{vw}$ and the potentials  $q(v)$ on finite number of edges and vertices, respectively.
 
The other topic of this paper concerns the Schr\"odinger operator on a  metric graph $\Gamma = \{\mathcal V, \mathcal E\}$, with vertex set $\mathcal V$ and edge set $\mathcal E$, and the topology determined by an appropriate adjacency matrix. The metric character of the graph means that each edge is identified with a line segment, in our case finite, and parametrized by its arclength. This makes it possible to endow
$\Gamma$ naturally with the metric defined as the length of the shortest path between two points. We do not fix the orientation of a given edge $e$, that is, the graph is undirected. We assume that for $v, v' \in \mathcal V$, there exists at most one edge with end points $v$, $v'$, and that $\Gamma$ has no loops.
 This can be assumed without loss of generality, since otherwise one can insert a `dummy' vertex of degree $2$ to any `superfluous' edge. With each edge $e \in \mathcal E$, we associate a one-dimensional Schr{\"o}dinger  operator
\begin{equation}
h_e := - \frac{d^2}{dz^2} + V_e(z), \quad z  \in [0,\ell_e] =: I_e,
\label{edge Schrodinger}
\end{equation}
where the length $\ell_e$ of the edge $e$ is a positive constant. To convert the collection of operators (\ref{edge Schrodinger}) into a self-adjoint Schr\"odinger operator on the whole graph, one has to impose conditions matching the functions at the vertices. In general, self-adjoint operators referring to the differential expression in question are parametrized by ${\rm deg}\, v \times {\rm deg}\, v$ unitary matrices, 
 cf.  \cite{KostrykinSchrader1999} or \cite{BerkolaikoKuchment2013}, Theorem 1.4.4.  If we require continuity of the functions at the vertices, however, this multitude is reduced to a one-parameter family, which we adopt in our case. To be concrete, for $\hat f \in H^2_{loc}(\mathcal E)$, we impose the generalized Kirchhoff condition, otherwise known as $\delta$-coupling: if $\hat f = \{\hat f_e\}_{e\in \mathcal E}$ such that $\hat f \in C(\Gamma)$ and $\hat f_e \in C^1(I_e)$, it holds that
\begin{equation}
\sum_{e \sim v} \hat f'_e(v) = C_v\hat f(v), \quad v \in \mathcal V,
\label{GeneralizedKirchhoff}
\end{equation}
where $f'_e(v)$ is given by (\ref{Defined/dzef(v)=f'e(v)}) and $e \sim v$ means that $v$ is an endpoint of the edge $e$, $C_v$ is a real constant,  $\hat f(v) = \hat f_e(0)$ if $e(0) = v$. Note that such a Hamiltonian can be defined as the norm-resolvent limit as $\kappa \to \infty$ of the following operators,
$$ 
\tilde h_{e,\kappa} = - \frac{d^2}{dz^2} + V_e(z) + \kappa W_{e}(\kappa z),
$$ 
with the usual Kirchhoff condition $\sum_{e \sim v}\hat f'_e(v) = 0$ for any $v \in \mathcal V$, where $C_v:= \sum_{e \sim v} \int_e W_e(z)\,dz$ and $W_{e}\in L^1(e)$ is a fixed function, cf. \cite{Exner96}. Note also that the singular vertex couplings with functions discontinuous at the vertex also allow for an interpretation, but the corresponding approximation procedure is considerably more complicated, see \cite{ChExTu10}.

We develop an inverse spectral and scattering theory associated with such quantum graphs which would facilitate a recovery of the graph structure, potentials $V_e(z)$, and constants $C_v$.
Roughly speaking, we consider a locally perturbed periodic graph, and prove the following result (Theorem \ref{qunatumgraphstrutureTheoremScattering}):
 
\begin{itemize}
\item Consider an infinite quantum graph $\Gamma = \{\mathcal V, \mathcal E\}$ on which all $\ell_e$, $V_e(z)$ coincide for all $e \in \mathcal E$, and $C_v/{\rm deg}\, v$ coincide for all $v \in \mathcal V$. If $\Gamma$ is a local perturbation of a periodic lattice of a certain class, then we can determine the graph structure of $\Gamma$ from the S-matrix for all energies.
\end{itemize}
Here a local perturbation of lattice means replacing a finite number of edges and vertices by a finite number of other edges and vertices.

The proof will be done by showing the equivalence of the S-matrix and the Dirichlet-to-Neumann (D-N) map in a bounded domain, and by reducing the problem to inverse problems for discrete Schr\"odinger operators of the type (\ref{Equationinafiitepart}).

 \subsection{Plan of the work}
 We  proceed in the following steps.
\begin{enumerate}
\item Preliminaries on metric graphs (\S \ref{Metricgraphsection}).
\item Inverse boundary value problem with the D-N map for a finite graph (\S \ref{SectionFinitegraphInverseproblem}): Use the results from \cite{BILL} to determine the structure of finite discrete graphs and quantum graphs from the knowledge of the corresponding D-N map.
\item  Inverse scattering for discrete Hamiltonians (\S \ref{InverseScatteringDiscrete}): Show that the S-matrix and the D-N map are equivalent and thus reduce the inverse scattering problem to the inverse boundary value problem.
 \item   Inverse scattering for quantum graphs (\S \ref{SectionQunatumGraphSpectralTheory}, \S \ref{InverseQuantumGraph}): Develop the spectral and scattering theory for locally perturbed periodic graph Laplacians, show that the S-matrix and the D-N map are equivalent, and recover the perturbations from the D-N map.
 \end{enumerate}

We end this section by the lists of assumptions and notations used in this paper except  standard ones.
 
 \begin{center}
 \begin{tabular}{|cc|cc|} \hline
 \multicolumn{4}{|c|}{\textbf{Assumptions}} \\ \hline
 {(M-1) - (M-5)} &{\S 2} &
  {(A-1) - (A-4)} &\S 2.3 in \cite{AIM18}   \\ \hline
  {(B-1) - (B-3)}  &\S 2.3 in \cite{AIM18} & 
  {(C-1), (C-2), (C-1)'} &\S \ref{s: The D-N maps}, \S \ref{Sec:3.2DiscreteGraphLaplacian} \\ \hline
  {(D-1) - (D-4) }
  & \S \ref{subsection-assumptions}  &   {(E-1)} & \S \ref{section-scatteringproof} \\ \hline  
  \end{tabular}
  \end{center}
  \vspace{4mm}

\hspace{-15mm}
 \begin{tabular}{|cc|cc|cc|cc|cc|} \hline
 \multicolumn{10}{|c|}{\textbf{Notations}} \\\hline
  $h_e$ &  (\ref{edge Schrodinger})  & $C_v$ &(\ref{GeneralizedKirchhoff})  & $d_v$ &(\ref{Definedv}) & $\phi_{e0}(z,\lambda)$   &  (\ref{phie0}) & $\phi_{e1}(z,\lambda)$ &(\ref{phie1}) \\ 
 $r_e(\lambda)$ & (\ref{Definerelambda})   & 
 $\widehat{\Delta}_{\mathcal V,\lambda}$  & (\ref{reducedvertexLaplacian})& $\widehat{Q}_{\mathcal V,\lambda}$&  (\ref{reducedvertexpotential})&$\widehat{T}_{\mathcal V,\lambda}$  & (\ref{DefineTVlambdaf(v)=1/dvsume(0)=vintIedz})&$\Lambda_{\mathcal V}(\lambda)$  & (\ref{S3DNforHmathcalVlambda}) \\ 
 $\Lambda_{\mathcal E}(\lambda)$  &(\ref{DefineLambdaElamnda}) &$\ell_{\mathcal E}, V_{\mathcal E}(z)$  & (\ref{S3EquilateralCond})& $\kappa_{\mathcal V}$&  (\ref{S3Equivertexpotentialcond})&$\mathcal U_{\mathcal V}$ & (\ref{DefineMathcalUmathcalV})& $\widehat{\Delta}_{\Gamma_0}$ &  (\ref{Laplaciandef-3})\\ 
$ \mathcal T_1$ &  (D-1) & $\mathcal T_0$  &  (D-2)&  $\widehat P_{ext}$& (\ref{Pextdefine}) & $\simeq$ & (\ref{Equationsimmomentum})&  $\Sigma$ & (\ref{DefineSigma}) \\ 
$E(\lambda)$ &(\ref{DefineE(lambda)}) & $\sigma^{(0)}(h^{(0)})$    & (\ref{Definsigma0h0})& $\sigma^{(0)}(- \widehat{\Delta}_{\mathcal V})$& (\ref{Definesigam0deltaV})  & $\sigma^{(0)}_{\mathcal T}$ & (\ref{Definesigma0tau})& $\mathcal T$& (\ref{Definetau0tau1sp0h0)}) \\ \hline
  \end{tabular}

\bigskip
The work of P.E. was supported by the Czech Science Foundation within the project 21-07129S and by the EU project ${\rm CZ}.02.1.01/0.0/0.0/16$\underline{ }$ 019/0000778$. H.I.  is supported by Grant-in-Aid for Scientific Research (C) 20K03667 Japan Society for the Promotion of Science. They are indebted to these supports.


\section{Metric graph and the associated discrete operator}
\label{Metricgraphsection}

Rephrasing the treatment of a Schr\"odinger operator, with or without a potential, on a \emph{metric} graph to the analogous problem on a \emph{combinatorial} (or discrete) graph is a well-known procedure that has been discussed in many papers, e.g. \cite{BoEgRu15,Cattaneo97, Exner97, Pankrashkin13}. We repeat it here mainly to fix notations.
Let $\Gamma = \{\mathcal V, \mathcal E\}$ be a metric graph with the vertex set $\mathcal V$ and edge set $\mathcal E$.
Note that for the metric graph, an edge $e \in \mathcal E$ is a segment between two vertices while for the discrete graph,  an edge is a pair of vertices. To avoid the complexity of notation, we use the same symbol $e_{vw}$  for an edge with endpoints $v, w$ for both graphs, often omitting  $v, w$. However, we will make a distinction between them in the arguments in \S2 {following} Definition \ref{DefinitionhatDeltaVlambda} and those in \S \ref{s: The D-N maps}. 
{For $v, w \in \mathcal V$, we say that $v$ and $w$ are adjacent, denoted by $v \sim w$, if there exists an edge having $v$ and $w$ as its endpoints.  For a subset 
$A \in \mathcal V$ or $\mathcal E$, $v \sim A$ and $A \sim v$ mean that $v$ is adjacent to some $w \in A\cap \mathcal V$. In particluar, for an edge $e \in \mathcal E$ and $v \in \mathcal V$, $e \sim v$ means that $v$ is an end point of $e$.}
The degree of a vertex $v \in \mathcal V$ is defined as
\begin{equation}
d_v := {\rm deg}\, v = \sharp\{e \in \mathcal E\, ; \, e \sim v\}.
\label{Definedv}
\end{equation}
Recall that for adjacent $v, v' \in \mathcal V$, 
the edge joining $v$ and $v'$ is unique by assumption.
For a function $\hat f = \{\hat f_e\}_{e\in\mathcal E}$ on $\Gamma$, with $\hat f_e : I_e \to {\bf C}$, and $e \in \mathcal E$ with $e \sim v$, we define
\begin{equation}
\frac{d \hat f}{d \nu_e}(v) := \hat f'_e(v).
\label{Defined/dzef(v)=f'e(v)}
\end{equation}
When computing the right-hand side, we parametrize $e$  as $e(z),\ z\in [0,\ell_e]$ with $e(0) = v$,
 and the boundary derivative is taken in the \emph{outward} direction with respect to $v$, see  \cite[Sec. I.4]{BerkolaikoKuchment2013}. Equivalently, the boundary derivatives can be written as
\begin{equation*}
\begin{split}
\frac{d \hat f}{d\nu_e}(e(0)) = \hat f_e'(0), \quad \frac{d \hat f}{d \nu_e}(e(\ell_e)) = - \hat f_e'(\ell_e),
\end{split}
\nonumber
\end{equation*}
where $\ell_e$ is the length of the edge $e$.
For the sake of brevity, we use the following shorthand notation:
$$ 
\int_e \hat u = \int_0^{\ell_e}
\hat u(z)dz.
$$ 
Then the following Green's formula holds:
$$ 
-\int_e(\hat u')'\hat w = \frac{d \hat u}{d \nu_e}\hat w\Big|_{e(0)} + \frac{d \hat u}{d \nu_e}\hat w\Big|_{e(\ell_e)} +
\int_e \hat u'\hat w'.
$$
For an edge $e \in \mathcal E$, let $L^2(e)$ be the set of all $L^2$-functions on $e$, conventionally understood as equivalence classes of functions coinciding a.e., and put
$$
L^2(\mathcal E) = {\mathop\bigoplus}_{e\in \mathcal E}\, L^2(e).
$$
For $\hat u = \{\hat u_e\}_{e\in\mathcal E}$ and $\hat w \in \{\hat w_e\}_{e\in\mathcal E}$, let $(\hat u,\hat w)_{\mathcal E}$ be the inner product:
$$
(\hat u,\hat w)_{\mathcal E} = \sum_{e\in\mathcal E}(\hat u_e,\hat w_e)_{e}
= \sum_{e\in\mathcal E}\int_e\hat u_e\overline{\hat w_e}.
$$
The Sobolev spaces are defined by
$$
H^m(\mathcal E) = {\mathop\bigoplus}_{e\in \mathcal E}\, H^m(e).
$$
Note that different conventions are used and sometimes the definition may involve the continuity at the vertices, see \cite[Def.~I.3.6]{BerkolaikoKuchment2013}.

Given a real-valued function $V_e \in L^1(e)$ on each $e \in \mathcal E$, we define a multiplication operator $V$ by
$$
\big(V\hat u\big)_e(z) = V_e(z)\hat u_e(z).
$$
Let $C_v$ be a real-valued function on $\mathcal V$. Throughout the paper we impose the following requirements:

\medskip
\noindent
(M-1) \hskip 20mm $ 0 < \inf_e\ell_e \leq \sup_e\ell_e < \infty, $

\smallskip
\noindent
(M-2) \hskip 20mm
$ \sup_{v \in \mathcal V}{d_v} < \infty,$

\smallskip
\noindent
(M-3) \hskip 20mm $ \sup_{e\in \mathcal E}\|V_e\|_{L^1(e)} < \infty, $

\smallskip
\noindent
(M-4) \hskip 20mm $ V_e(z) = V_e(\ell_e - z)$, 

\smallskip
\noindent
(M-5) \hskip 20mm $ \sup_{v \in \mathcal V} |C_v| < \infty$.

\medskip
\noindent
Naturally all of these requirements except the symmetry condition (M-4) are satisfied automatically if the graph $\Gamma$ is finite. We define the operator $\widehat H_{\mathcal E}$ by 
\begin{equation} \label{def-eLaplace}
\big(\widehat H_{\mathcal E}\hat u\big)_e(z)  = - \hat u''_e(z) + V_e(z)\hat u_e(z)
\end{equation}
acting on $I_e$, with the domain consisting of functions
 \begin{equation}
 \hat u \in D(\widehat H_{\mathcal E}) \Longleftrightarrow
 \left\{
 \begin{split}
 &  \hat u  \in H^2(\mathcal E), \quad \hat u \in C(\Gamma),\\
 & \sum_{e\sim v}\hat u'_e(v) =   C_v\hat u(v),  \quad
 v \in \mathcal V.
 \end{split}
 \right.
\label{S2DomainH}
 \end{equation}
Here in the first line of the right-hand side, $\hat u \in C(\Gamma)$ means that $\hat u_e(v) = \hat u_{e'}(v)$ if $v \sim e$, $v \sim e'$ and that 
 $\hat u$, thus defined globally on $\mathcal E$, is continuous  {on the whole graph $\Gamma$.
It is straightforward to check that $\widehat H_{\mathcal E}$ is self-adjoint.

Let $\lambda \in \mathbb{C}\setminus\mathbb{R}$. For any  edge $e \in \mathcal E$, let $\phi_{e0}(z,\lambda)$  and $\phi_{e1}(z,\lambda)$ be the solutions of $- \phi'' + V_e\phi = \lambda \phi$ on $I_e$ satisfying the boundary conditions
\begin{align}
 \phi_{e0}(0,\lambda) = 0,& \quad  \phi_{e0}'(0,\lambda) = 1, 
 \label{phie0}\\[.5em]
 \phi_{e1}(\ell_e,\lambda) = 0,& \quad  \phi_{e1}'(\ell_e,\lambda) = -1.
 \label{phie1}
 \end{align}
Note that $\phi_{e1}(z,\lambda) = \phi_{e0}(\ell_e-z,\lambda)$ by the symmetry condition (M-4). Let $r_e(\lambda)$ be the Green operator of $- d^2/dz^2 + V_e(z) - \lambda$ on $e$ with the Dirichlet boundary condition:
\begin{equation}
\begin{split}
r_e(\lambda)\hat f_e & = \Big(- \frac{d^2}{dz^2} + V_e(z) - \lambda\Big)^{-1}\hat f_e  = \int_{I_e}r_e(z,z',\lambda)\hat f_e(z')dz',
\end{split}
\label{Definerelambda}
\end{equation}
where the integral kernel is given by
\begin{align*}
r_e(z,z',\lambda) = & - \frac{1}{W_e(z',\lambda)}\left\{
\begin{array}{ll}
\phi_{e0}(z,\lambda)\phi_{e1}(z',\lambda), & \quad 0 < z < z', \\
\phi_{e1}(z,\lambda)\phi_{e0}(z',\lambda), & \quad 0 < z' < z,
\end{array}
\right. \\[.5em]
W_e(z,\lambda) =  &\: \phi_{e0}(z,\lambda)\phi_{e1}'(z,\lambda) - \phi_{e0}'(z,\lambda)\phi_{e1}(z,\lambda).
\end{align*}
Let $\hat u = (\widehat H_{\mathcal E} - \lambda)^{-1}\hat f$. Then on each edge $e$, the function $\hat u_e(z,\lambda)$ can be written as
 \begin{equation}
 \hat u_e(z,\lambda) = c_e(\ell_e,\lambda)\frac{\phi_{e0}(z,\lambda)}{\phi_{e0}(\ell_e,\lambda)} + c_e(0,\lambda)\frac{\phi_{e1}(z,\lambda)}{\phi_{e1}(0,\lambda)}
 + r_e(\lambda)\hat f_e,
 \label{S2Formulahatuezlambda}
 \end{equation}
where the constants $c_e(\ell_e,\lambda), c_e(0,\lambda)$ are determined by the $\delta$-coupling condition (\ref{GeneralizedKirchhoff}). Since
$\phi_{e0}'(0,\lambda) = 1$ and $\phi_{e1}'(0,\lambda) = - \phi_{e0}'(\ell_e,\lambda)$, we infer that
\begin{equation*}
\begin{split}
\frac{d}{dz}r_e(\lambda)\hat f_e\Big|_{z=0} & = -
\int_{I_e}
\frac{\phi_{e1}(z',\lambda)}{W_e(z',\lambda)}\hat f_e(z')\,dz',
\end{split}
\nonumber
\end{equation*}
and consequently we have
\begin{equation*}
\hat u_e'(0,\lambda) = \frac{1}{\phi_{e0}(\ell_e,\lambda)}\Big(c_e(\ell_e,\lambda) - \phi'_{e0}(\ell_e,\lambda)c_e(0,\lambda)\Big) - \int_{I_e}\frac{\phi_{e1}(z',\lambda)}{W_e(z',\lambda)}\hat f_e(z')\,dz'.
\nonumber
\end{equation*}
Since $\hat u_e(0,\lambda)=c_e(0,\lambda)$, the $\delta$-coupling condition (\ref{GeneralizedKirchhoff}) can be rewritten as
\begin{equation}
\begin{split}
\sum_{e(0) = v}
& \left(\frac{1}{\phi_{e0}(\ell_e,\lambda)}
\Big(c_e(\ell_e,\lambda) - \phi_{e0}'(\ell_e,\lambda)c_e(0,\lambda)\Big) - \frac{C_v}{d_v}c_e(0,\lambda)\right)\\
& = \sum_{e(0) = v}
\int_{I_e}\frac{\phi_{e1}(z',\lambda)}{W_e(z',\lambda)}\hat f_e(z')\,dz'.
\label{S2deltacouplingrewritten}
\end{split}
\end{equation}
To make the dependence on the edge parametrization more visible, we alternatively write $\hat f_e(e(z))$ instead of a function $\hat f_e(z)$ on $I_e$.

From here until the end of \S \ref{s: The D-N maps}, we distinguish the edges in the metric graph and those of the discrete graph, denoting the edges and the functions on the former by $\underline e$, $\hat{\underline u}$, ${\hat{\underline{u}}}_{\underline e}$, and those for the discrete graph by $e$, $\hat u$ and $\hat u_e$.

 \begin{definition}
 \label{DefinitionhatDeltaVlambda}
The weighted discrete graph Laplacian 
$\widehat\Delta_{\mathcal V,\lambda} :  \ell^2(\mathcal V) \to \ell^2(\mathcal V)$, where  $\ell^2(\mathcal V) =  {{\mathbb C}}^{\sharp\mathcal V}$,  on $\mathcal V$, associated with the Schr\"odinger operator on $\Gamma$ specified by \eqref{edge Schrodinger} and \eqref{GeneralizedKirchhoff}, acts on a function $\hat u(v)$ on $\mathcal V$ as
 \begin{equation}
 \begin{split}
 \big(\widehat\Delta_{\mathcal V,\lambda}\hat u\big)(v) & =  \frac{1}{d_v}\sum_{e(0) = v,\, e\in \mathcal{E}}\frac{1}{\phi_{e0}(\ell_e,\lambda)} \hat u(e(\ell_e))  \\
  & = \frac{1}{d_v}\sum_{w\sim v,\, w\in \mathcal{V}}\frac{1}{\phi_{e0}(w,\lambda)} \hat u(w).
  \end{split}
  \label{reducedvertexLaplacian}
 \end{equation}
We introduce the discrete scalar potential $\widehat Q_{\mathcal V,\lambda} = \{\widehat Q_{v,\lambda}\}_{v \in \mathcal V}$ by
 \begin{equation}
 \widehat Q_{v,\lambda} =  \frac{1}{d_v}\sum_{e \sim v,\, e\in \mathcal{E}}\frac{\phi_{e0}'(\ell_e,\lambda)}{\phi_{e0}(\ell_e,\lambda)} + \frac{C_v}{d_v}.
 \label{reducedvertexpotential}
 \end{equation}
Note that $e(0) = v$ and $e(\ell_e) = w$ hold in the definitions (\ref{reducedvertexLaplacian}) and (\ref{reducedvertexpotential}).
  \end{definition}

Furthermore, defining
 \begin{equation}
 \big(\widehat T_{\mathcal V,\lambda}\underline{\hat f}\big)(v) :=
 \frac{1}{d_v}\sum_{\underline e(0) = v}\int_{I_{\underline e}}\frac{\phi_{\underline e 0}(z,\lambda)}{\phi_{\underline e 0}(\ell_{\underline e},\lambda)}\underline{\hat f}_{\underline e}(z)\,dz,
\label{DefineTVlambdaf(v)=1/dvsume(0)=vintIedz}
 \end{equation}
 we can rewrite the coupling condition (\ref{S2deltacouplingrewritten}) in the following way.

 \begin{lemma}
 \label{S2LemmaDeltamathcalVmathcalV=Tlambda}
 The $\delta$-coupling condition (\ref{GeneralizedKirchhoff}) can be expressed as
 \begin{equation}
 \left(- \widehat\Delta_{\mathcal V,\lambda} +
 \widehat Q_{\mathcal V,\lambda}\right)\hat{\underline u}(v) = \widehat T_{\mathcal V,\lambda}\underline{\hat f}(v), \quad v \in \mathcal V.
 \label{EqinLemma2.2}
 \end{equation}
 \end{lemma}

\medskip

\noindent Assuming that the equation (\ref{EqinLemma2.2}) is solvable, we write $\hat {\underline u} = \{\hat{\underline  u}_{\underline e}\}_{{\underline e}\in \mathcal E}$ in the form of (\ref{S2Formulahatuezlambda}) with
$c_{\underline e}(0,\lambda)$, $c_{\underline e}(\ell_{\underline e},\lambda)$ being the vertex values of $\hat {\underline u}(v)$ at $v = {\underline e}(0)$ and $v = {\underline e}(\ell_{\underline e})$, respectively. Then we have
 $$ 
 \hat {\underline u}\big|_{\mathcal V} = \big(- \widehat\Delta_{\mathcal V,\lambda} +
 \widehat Q_{\mathcal V,\lambda}\big)^{-1}\widehat T_{\mathcal V,\lambda}\underline{\hat f}.
 $$ 
 Note further that the adjoint operator $(\widehat T_{\mathcal V,\lambda})^{\ast}$ acts as
\begin{equation}
 \begin{split}
\big( (\widehat T_{\mathcal V,\lambda})^{\ast}\,\hat{\underline g}\big)_{\underline e}(z) & =
 \sum_{v = {\underline e}(0)}\frac{1}{d_v}\frac{\phi_{{\underline e}0}(z,\overline{\lambda})}{\phi_{{\underline e}0}(\ell_{\underline e},\overline{\lambda})}\,\underline{\hat g}(v)  \\
 & =  \frac{1}{d_{{\underline e}(0)}}\frac{\phi_{{\underline e}0}(z,\overline{\lambda})}{\phi_{{\underline e}0}(\ell_{\underline e},\overline{\lambda})}\underline{\hat g}({\underline e}(0)) + \frac{1}{d_{{\underline e}(\ell_{\underline e})}}
 \frac{\phi_{{\underline e}1}(z,\overline{\lambda})}{\phi_{{\underline e}1}(0,\overline{\lambda})}\,\underline{\hat g}({\underline e}(\ell_{\underline e})),
 \end{split}
\label{DefineTVlambdaasthatg=frac1dVphi/phig(0)}
 \end{equation}
where in the first line we consider both orientations of the edge $\underline e$, while in the second line we fix one orientation. Now we define the operator $r_{\mathcal E}(\lambda)$ on $\mathcal{E}$ by
 $$ 
 r_{\mathcal E}(\lambda)\underline{\hat f} = r_{\underline e}(\lambda)\underline{\hat f}_{\underline e} \quad \text{on}\;\; {\underline e},
 $$ 
and we arrive at the following Krein-type formula expressing the resolvent through its comparison to that of the Dirichlet-decoupled graph.

\begin{lemma}
\label{S2LemmaRElambdaformula}
The resolvent $\widehat R_{\mathcal E}(\lambda) = (\widehat H_{\mathcal E} - \lambda)^{-1}$ is expressed as
$$ 
\widehat R_{\mathcal E}(\lambda) = (\widehat T_{\mathcal V, \overline{\lambda}})^{\ast}\big(- \widehat\Delta_{\mathcal V,\lambda} +
 \widehat Q_{\mathcal V,\lambda}\big)^{-1}\widehat T_{\mathcal V,\lambda} + r_{\mathcal E}(\lambda).
$$ 
\end{lemma}

\medskip

\noindent Let us note here that for $\lambda \notin \mathbb{R}$, the coefficients of $\widehat\Delta_{\mathcal V,\lambda}$ and $\widehat Q_{\mathcal V,\lambda}$ are not real and hence the existence of the inverse $(- \widehat\Delta_{\mathcal V, \lambda} + \widehat Q_{\mathcal V, \lambda})^{-1}$ is not obvious. We postpone its justification until \S \ref{SectionQunatumGraphSpectralTheory}, and admit Lemma \ref{S2LemmaRElambdaformula} as a formal formula for the moment.


\section{Inverse boundary value problem for a finite graph}
\label{SectionFinitegraphInverseproblem}
\subsection{The D-N maps}
\label{s: The D-N maps}

 In this section, we consider a finite graph $\Gamma = \{\mathcal V,\mathcal E\}$ with boundary $\partial\mathcal V$ and assume that

  \medskip
 \noindent
\textbf{(C-1)}\ \ {$\Gamma$ consists of two parts called boundary $\partial\mathcal V$ and interior $\mathcal V^{o}$ whose vertex sets are disjoint; each boundary vertex is connected to only one interior vertex. }

\medskip
\noindent Note that, topologically speaking, the notion of the graph boundary is not trivial;
here we use the freedom to determine it \emph{ad hoc} to suit our purposes.
 
Let $\widehat{H}_{\mathcal{E}}$ be the quantum graph Schr{\"o}dinger operator on the finite graph $\Gamma$  as in the previous section with Dirichlet boundary condition on the boundary $\partial \mathcal V$. 
We put
 \begin{equation}
{\sigma'} :=  \Big(\bigcup_{{\underline e}\in \mathcal E}\sigma(h_{\underline e})\Big) \cup
 \left\{{\lambda \in {\mathbb C}}\, ; \, \det(- \widehat \Delta_{\mathcal V,\lambda} + \widehat Q_{\mathcal V,\lambda}) = 0\right\},
 \label{DefineSigma}
 \end{equation}
which is discrete in ${\mathbb C}$, as $\Gamma$ is a finite graph.
Note $\sigma(\widehat H_{\mathcal E}) \subset \sigma'$. Let $ h_{\underline e}$ be the differential operator on $\underline e$ as in (\ref{edge Schrodinger}). Then
 for any $\lambda \not\in \sigma(\widehat H_{\mathcal E})$ and given boundary data $\underline{\hat f}$,  there is a unique solution $\hat{\underline u} = \{\hat{\underline u}_{\underline e}\}_{{\underline e}\in\mathcal E}$ to the equation
 \begin{equation}
 \left\{
 \begin{split}
& (h_{\underline e} - \lambda)\hat{\underline u}_e = 0 \quad {\rm on} \quad  \forall {\underline e} \in \mathcal E,\\
&
 \hat{\underline u} = \underline{\hat f} \quad {\rm on} \quad \partial\mathcal V, \\
 &  \delta {\text {\rm -coupling  condition}}\   (\ref{GeneralizedKirchhoff}).
 \end{split}
 \right.
 \label{S3BVPgraph}
\end{equation}
Here, as in (\ref{S2DomainH}), $\hat{\underline u}$ is assumed  to be in $C(\Gamma)$.
Using the solution $\hat{\underline u}$, we define the D-N map $\Lambda_{\mathcal E}(\lambda) : \mathbb{C}^m \to \mathbb{C}^m,\: m = \sharp\partial\mathcal V$, by
 \begin{equation}
 \Lambda_{\mathcal E}(\lambda) : \underline{\hat f} \to \hat{\underline u}'_{\underline e}(v), \quad {\underline e}(0) = v \in \partial\mathcal V.
 \label{DefineLambdaElamnda}
 \end{equation}
Note that $\hat{\underline u} = \{\hat{\underline  u}_e\}_{{\underline e}\in \mathcal E}$ is the solution to the edge Schr{\"o}dinger equation (\ref{S3BVPgraph}) if and only if $\hat{\underline u}\big|_{\mathcal V}$ is the solution to the vertex Schr{\"o}dinger equation (\ref{BVPDiscrete}) ( cf. \cite{Exner97}).

Under the  Dirichlet boundary condition on the boundary $\underline e(0)$ and $\underline e(\ell_{\underline e})$,  $h_{\underline e}$ has
discrete spectrum, and for any $\lambda \not\in \cup_{{\underline e}\in \mathcal E}\sigma(h_{\underline e})$, we have $\phi_{{\underline e}0}(\ell_{\underline e},\lambda) \neq 0$. Hence the weighted discrete Laplacian (\ref{reducedvertexLaplacian}) is well defined.
We consider the boundary value problem for the corresponding Schr\"odinger-type operator $\widehat H_{\mathcal V,\lambda} =: - \widehat \Delta_{\mathcal V,\lambda} + \widehat Q_{\mathcal V,\lambda}$ on the vertex set $\mathcal V$ with the boundary value $\hat f$ on $\partial\mathcal V$, namely
  \begin{equation}
  \left\{
  \begin{split}
& \left(- \widehat\Delta_{\mathcal V,\lambda} +
 \widehat Q_{\mathcal V,\lambda}\right)\hat u(v) = 0, \quad v \in \mathcal V^{o} = \mathcal V\setminus\partial\mathcal V,\\
 & \hat u(v) = \hat f(v), \quad v \in \partial\mathcal V.
 \end{split}
 \right.
 \label{BVPDiscrete}
 \end{equation}
Using the solution $\hat u_{\mathcal V}$, which depends also on $\lambda$ and is denoted by $\hat u_{\mathcal V}(v,\lambda)$, we next define the D-N map for $\widehat H_{\mathcal V,\lambda}  : \mathbb{C}^m \to \mathbb{C}^m$ by
\begin{equation}
\Lambda_{\mathcal V}(\lambda) : \hat f \to
\frac{1}{\phi_{e0}(w,\lambda)} {\hat u_{\mathcal V}(w,\lambda)}, \quad
w = e(\ell_e), \quad v =  e(0) \in \partial\mathcal V.
\label{S3DNforHmathcalVlambda}
\end{equation}
 
Therefore, $\Lambda_{\mathcal E}(\lambda)$ and $\Lambda_{\mathcal V}(\lambda)$ are meromorphic functions of $\lambda$ with poles in the discrete set $\Sigma$. Recall that given a subset $A \subset \mathcal V$ and an edge $e \in \mathcal E$ {(or $\underline e$)}, we say that $e$ is adjacent to $A$, denoted as $e \sim A$ or $A \sim e$, if $e(0) \in A$ and $e(\ell_e) \not\in A$.

 \begin{lemma} \label{edge-vertex-equiv}
 \label{DNmathcalE=DNmathcalV}
 Assuming that we know $\ell_{\underline e}$ and $V_{\underline e}(z)$ for all $\underline e$ adjacent to $\partial\mathcal V$, then
 $\Lambda_{\mathcal E}(\lambda)$ and $\Lambda_{\mathcal V}(\lambda)$ determine each other for any $\lambda \not\in \sigma'$.
 \end{lemma}
\begin{proof}
Given the solution $\hat{\underline u}$  to (\ref{S3BVPgraph}), the corresponding $\hat{\underline  u}\big|_{\mathcal V}$ solves (\ref{BVPDiscrete}). Conversely, given the solution $\hat u_{\mathcal V}$ of (\ref{BVPDiscrete}), we define $\hat{\underline u}$ by
$$
\hat u_{\underline e}(z) = c_{\underline e}(\ell_{\underline e},\lambda)\frac{\phi_{{\underline e}0}(z,\lambda)}{\phi_{{\underline e}0}(\ell_{\underline e},\lambda)} + c_{\underline e}(0,\lambda)\frac{\phi_{{\underline e}1}(z,\lambda)}{\phi_{{\underline e}1}(0,\lambda)},
$$
where on the edge with the initial vertex $v = {\underline e}(0) \in \partial \mathcal V$, we put
$$
c_{\underline e}(0,\lambda) = \underline{\hat f}(v).
$$ 
The function $\hat u$ defined in this way solves (\ref{S3BVPgraph}). The D-N map for $\widehat H_{\mathcal E}$ is
$$
 \Lambda_{\mathcal E}(\lambda) :
\underline{\hat f} \to c_{\underline e}(\ell_{\underline e},\lambda)\frac{1}{\phi_{{\underline e}0}(\ell_{\underline e},\lambda)} +
\underline{\hat f}(v)\frac{\phi_{{\underline e}1}'(0,\lambda)}{\phi_{{\underline e}1}(0,\lambda)}, \quad v = {\underline e}(0) \in \partial\mathcal V.
$$
The D-N map for $\widehat H_{\mathcal V,\lambda}$ is, by (\ref{S3DNforHmathcalVlambda}), {taking $w = {\underline e}(\ell_{\underline e})$}, 
\begin{equation}
 \Lambda_{\mathcal V}(\lambda) :
\underline{\hat f} \to \frac{1}{\phi_{{\underline e}0}(\ell_{\underline e},\lambda)}
\left(
c_{\underline e}(\ell_{\underline e},\lambda) +
\frac{\underline{\hat f}(v)}{\phi_{{\underline e}1}(0,\lambda)}
\right).
\label{S2DNmapforwidehatDeltamathcalVlambda}
\end{equation}
Since we know $\phi_{{\underline e}0}(z,\lambda)$, $\phi_{{\underline e}1}(z,\lambda)$ for edges $\underline e$ adjacent to $\partial\mathcal V$, the knowledge of the D-N maps for both the $\widehat H_{\mathcal E}$ and $\widehat H_{\mathcal V,\lambda}$ is thus equivalent to that of the initial value problem or the two-point boundary value problem for $h_{\underline e} - \lambda$ on each edge $\underline e$. Consequently, the two D-N maps are equivalent.
\end{proof}

\subsection{A reminder: inverse problem for the discrete graph Laplacian.}
\label{Sec:3.2DiscreteGraphLaplacian}
To make this paper self-contained, let us recall a result obtained in \cite{BILL} as follows.  
 We say that the collection $\mathbb{G}= \{G, \partial G, E, \mu, g\}$ is a weighted discrete graph with boundary, if it satisfies the following conditions.
\begin{itemize}
\item {\it $\{G\cup \partial G,E\}$ is an undirected simple discrete graph, where $G \cup \partial G$ is the set of vertices and $E$ is the set of edges. Assume that $G\cap\partial G = \emptyset$. We call $G$ the interior of the graph, and call $\partial G$ the boundary of the graph. }
\item {\it $\mu : G \cup \partial G \to \mathbb{R}_+$ is a weight function on vertices.}
\item {\it $g : E \to \mathbb{R}_+$ is a weight function on edges.}
\end{itemize}
We say $\mathbb{G}$ is finite (resp. connected) if $\{G\cup \partial G,E\}$ is finite (resp. connected).
When the weights $\mu,g$ are not relevant in a specific context, we write $\{G,\partial G,E\}$ for short.
In \S \ref{Sec:3.2DiscreteGraphLaplacian} and \S \ref{section-scatteringproof}, we use $x, y, z$ to
refer {to} vertices in $G$.

Given a subset $S \subset G$, we say that $x_0 \in S$ is an {\it extreme point} of $S$ with respect to $\partial G$ if
$$
\exists z \in \partial G  \textrm{ such that }  d(x_0,z) < d(x,z), \ \forall x \in S, \ x \neq x_0,
$$
where $d(x,y)$ is the distance of $x, y \in G\cup\partial G$ understood as the minimum number of edges forming a path connecting the two points $x,y$. The following {\it Two-Points Condition} for $\{G,\partial G,E\}$ is imposed:

\medskip
\noindent
\textbf{(C-2)} {\it For any subset $S \subset G$ with $\sharp S \geq 2$, there exist at least two extreme points of $S$ with respect to $\partial G$. }

\medskip

We consider the set of points adjacent to the boundary defined as
$$ 
N(\partial G) = \{x \in G\, ; \, \exists z \in \partial G, \textrm{ such that } x \sim z\} \cup \partial G.
$$ 
We say that two weighted graphs with boundary $\mathbb{G}$, $\mathbb{G}'$ are {\it boundary isomorphic} if there exists a bijection $\Phi_0 : N(\partial G) \to N(\partial G')$ with the following properties.

\smallskip
(i) $\Phi_0\big|_{\partial G} : \partial G \to \partial G'$ is bijective.

  \smallskip
  (ii) For any $z \in \partial G,\, y \in N(\partial G)$ the equivalence
  $ y \sim z \Longleftrightarrow \Phi_0(y) \sim'\Phi_0(z)$ holds.

\smallskip
\noindent
  The graph Laplacian $\Delta_G$ is defined by
  $$ 
  \left(\Delta_Gu\right)(x) = \frac{1}{\mu_x}\sum_{y\sim x, y \in G\cup \partial G}g_{xy}\big(u(y) - u(x)\big),
  x \in G,
  $$ 
 and the Neumann derivative at the boundary is defined by
  \begin{equation}
  \left(\partial_{\nu}u\right)(z) = \frac{1}{\mu_z}\sum_{x\sim z, x \in G}g_{xz}\big(u(x) - u(z)\big), \quad
  z \in \partial G.
\label{S3partialnudefine}
  \end{equation}
Moreover, adding a potential function $q$ on $G$ to $\Delta_G$, we can define the D-N map in the same way as in the previous section.

\smallskip
  The following result is valid, cf. Theorems~1 and 2 in \cite{BILL}.

\begin{theorem}
\label{TheoremBILL}
Let $\mathbb{G} = \{G, \partial G, E, \mu, g\}$ and $\mathbb{G}' = \{G', \partial G', E', \mu', g'\}$ be two finite weighted graphs with boundary satisfying (C-1), (C-2), and let $q, q'$ be real-valued potential functions on $G, G'$. Suppose $\mathbb{G}$ and $\mathbb{G}'$ are boundary isomorphic via $\Phi_0$, and their D-N maps coincide for all energies. Then, there exists a bijection $\Phi : G \cup \partial G \to G'\cup\partial G'$ such that

\smallskip
\noindent
(1) \ $\Phi\big|_{\partial G} = \Phi_0\big|_{\partial G}$.

\smallskip
\noindent
(2) \
$x \sim y \Longleftrightarrow \Phi(x) \sim' \Phi(y)$, $\quad \forall\, x, y \in G\cup \partial G$,

\smallskip
\noindent
where $x' \sim' y'$ means that $x', y'$ are adjacent in $G'\cup\partial G'$.

\smallskip
Identifying vertices of $\mathbb{G}$ with those of $\mathbb{G}'$ by this bijection, assume furthermore that $\mu_z = \mu_z'$, $g_{xz} = g'_{xz}$ for all $z \in \partial G$, $x \in G$. Then we have

\smallskip
\noindent
(3) If $\mu = \mu'$, then $g = g'$, $q = q'$.

\smallskip
\noindent
(4) If $q = q' = 0$, then $\mu = \mu'$ and $g = g'$.

\smallskip
\noindent
 In particular, if $\mu(v) = {\rm deg}\,v$ and $\mu'(v') = {\rm deg}\,v'$ holds for all $v\in G$ and $v'\in G'$, respectively, then $g = g'$, $q = q'$.
\end{theorem}

\begin{remark}
Let us add three remarks. \\
\noindent
(1)  The theorems in \cite{BILL} that we refer to were formulated in terms of Neumann boundary spectral data; however, the claims hold for the Dirichlet boundary spectral data as well with a minor modification of the proof.
 \\
 \noindent
 (2) Under the conditions (C-1), (C-2), the Neumann boundary spectral data determine the N-D maps for all energies, that is, the N-D map of $-\Delta_G - \lambda$ for all $\lambda$,  and vice versa, see Lemma \ref{ND} below.
In the same way, the Dirichlet boundary spectral data and the D-N maps for all energies determine each other.
\\
\noindent
(3) We can replace the assumption (C-1) by 

\smallskip
\noindent
{\bf (C-1)'} For any $z \in \partial G$ and any $x, y \in G$, if $ x \sim z$, $y \sim z$, then $x \sim y$.

\smallskip
\noindent
cf. \cite{BILL}. Inspecting Figures \ref{Periodic hexagonal lattice} -- \ref{Perturbed triangular lattice} in \S \ref{InverseScatteringDiscrete} below, we see that (C-1) is satisfied for the hexagonal lattice, but not, e.g., for the triangular lattice. The latter, however, is covered by (C-1)'.  All the arguments below work under the assumption (C-1)'
with minor modification. For the sake of simplicity, we adopt (C-1) in this paper.
\end{remark}

                                                                                                                                                                                                                                                                                                                                                                                                                                                                                                                                                                                                                                                                                                                                                                                                                                                                                                                                                                                                                                                                                                                                                                                                                                                                                                                                                                                                                                \section{Equilateral graphs}
Suppose we are given a {finite} quantum graph $\Gamma = \{\mathcal V, \mathcal E\}$ satisfying
 (C-1), (C-2). We further assume that there exist a number $\ell_{\mathcal E}$ and a function $V_{\mathcal E}(z)$ such that
\begin{equation}
\ell_e = \ell_{\mathcal E}, \quad V_e(z) = V_{\mathcal E}(z), \quad \forall e \in \mathcal E.
\label{S3EquilateralCond}
\end{equation}
Moreover, assume that
\begin{equation}
k_{\mathcal V} := \frac{C_v}{d_v} \ {\rm is}  \ {\rm independent} \ {\rm of} \ v \in \mathcal V.
\label{S3Equivertexpotentialcond}
\end{equation}
Let $\phi_{0}(z,\lambda)$ and $\phi_{1}(z,\lambda)$ be $\phi_{e0}(z,\lambda)$, $\phi_{e1}(z,\lambda)$  in \S  \ref{Metricgraphsection}. {By (\ref{reducedvertexLaplacian}) and  (\ref{reducedvertexpotential})},  the discrete graph Laplacian $\widehat\Delta_{\mathcal V,\lambda}$ and the vertex potential $\widehat{Q}_{\mathcal V,\lambda}$ can be rewritten as 
\begin{equation} 
\big(\widehat\Delta_{\mathcal V,\lambda}\hat u\big)(v) = \frac{1}{d_v} \frac{1}{\phi_0(\ell_{\mathcal E},\lambda)}
\sum_{w \sim v}\hat u(w), \quad v \in \mathcal V,
\label{S3DNforequilateralgraph}
\end{equation}
\begin{equation} \label{potential-equilateral}
\widehat{Q}_{\mathcal V,\lambda} = \frac{1}{\phi_0(\ell_{\mathcal E},\lambda)}E_{\mathcal E}(\lambda),
\quad
E_{\mathcal E}(\lambda) = \phi_0'(\ell_{\mathcal E},\lambda) + k_{\mathcal V}\phi_0(\ell_{\mathcal E},\lambda).
\end{equation}
Thus \eqref{S3DNforequilateralgraph} and \eqref{potential-equilateral} differ by a multiplicative constant $\phi_0(\ell_{\mathcal E},\lambda)$ from the discrete operator with the graph Laplacian $(\widehat \Delta_{\mathcal V}\hat u)(v):= \frac{1}{d_v}\sum_{w\sim v} \hat u(w)$ and potential $E_ {\mathcal E}(\lambda)$. 

This amounts to considering  a graph $\widetilde{\Gamma}$ with the same edge set $\mathcal E$ and the vertex set $\mathcal V$ as our original $\Gamma$, and 
{$\mu_v = d_v$, $g_{vw}=1$.}
We let $\lambda$ vary and use analytic continuation: if we are given the D-N map for the original quantum graph $\Gamma$ for all energies, we can obtain the D-N map of the above discrete operator $\widehat \Delta_{\mathcal V}$ for all energies, and, \emph{mutatis mutandis}, the Dirichlet boundary spectral data for $\widehat \Delta_{\mathcal V}$ under the conditions (C-1), (C-2). Note that the D-N map for the operator $\widehat \Delta_{\mathcal V}$ acts as $\hat u(v) {\mapsto} \hat u(w)$, $\,w \sim v \in \partial{\mathcal V}$, $w \in \mathcal V$; by (\ref{S2DNmapforwidehatDeltamathcalVlambda}) it can be computed from the D-N map of $\widehat{\Delta}_{\mathcal V,\lambda}$ if we know $\phi_0(z,\lambda)$, i.e. $\ell_{\mathcal E}$ and $V_{\mathcal E}(z)$.

Suppose now that we are given two such graphs $\widetilde\Gamma = \{\mathcal V, \mathcal E\}$ and $\widetilde\Gamma' = \{\mathcal V',\mathcal E'\}$. Applying then Theorem \ref{TheoremBILL} with $\mu_x = d_x$, $g_{xy} = 1$, we infer that there is a bijection
$$ 
\Phi : \widetilde\Gamma \to \widetilde\Gamma'
\label{S3PhiBijection}
$$ 
preserving the graph structure. Setting $v' = \Phi(v)$, we conclude that
$$ 
d_v = d_{v'}, \quad \forall v \in \mathcal V,
$$ 
and consequently
$$ 
C_v = C'_{v'}, \quad \forall v \in \mathcal V.
$$ 
In this way, we have proven the following theorem:
\begin{theorem}
\label{qunatumgraphstrutureTheoremBVP}
Let   $\Gamma = \{\mathcal V, \mathcal E\}$ and $\Gamma' = \{\mathcal V',\mathcal E'\}$ be two {finite} quantum graphs satisfying assumptions (C-1), (C-2), (\ref{S3EquilateralCond}), (\ref{S3Equivertexpotentialcond}) and $\ell_{\mathcal E} = \ell_{\mathcal E'}, V_{\mathcal E}(z) = V_{\mathcal E'}(z)$,  $k_{\mathcal V}= k_{\mathcal V'}$. 
Suppose that the D-N maps for the Schr{\"o}dinger operator for the two quantum graphs coincide for all energies.
Then there is a bijection $\Phi : \Gamma \to \Gamma'$ preserving the graph structure, and $d_v = d_{v'}$, $C_v = C'_{v'}$ hold for all $v \in \mathcal V$ and $v' = \Phi(v)$.
\end{theorem}


\section{Inverse scattering for the discrete Hamiltonian}
\label{InverseScatteringDiscrete}

It is known that the potential of the discrete Schr\"odinger operator on periodic square 
or hexagonal lattices
can be uniquely recovered from the knowledge of the scattering matrix of all energies, see \cite{A1,IK}. Furthermore, the forward and inverse scattering problems have been considered for infinite graphs that are local perturbations of periodic lattices in \cite{AIM16,AIM18}. For several standard types of lattices, it was shown in \cite{AIM18} that the scattering matrix for the discrete Schr\"odinger operator on locally perturbed lattices determines the Dirichlet-to-Neumann map for the discrete Schr\"odinger equation on the perturbed subgraph. In this section, we apply Theorem \ref{TheoremBILL} to recover the potential on locally perturbed lattices, as well as to recover the structure of the perturbed subgraph (see Theorem \ref{mainscattering}). This result may be applied, in particular, to probe graphene defects from the knowledge of the scattering matrix, see Figures \ref{Periodic hexagonal lattice} and \ref{Perturbed hexagonal lattice}.

\begin{figure}[hbtp]
\centering
\includegraphics[width=5cm]{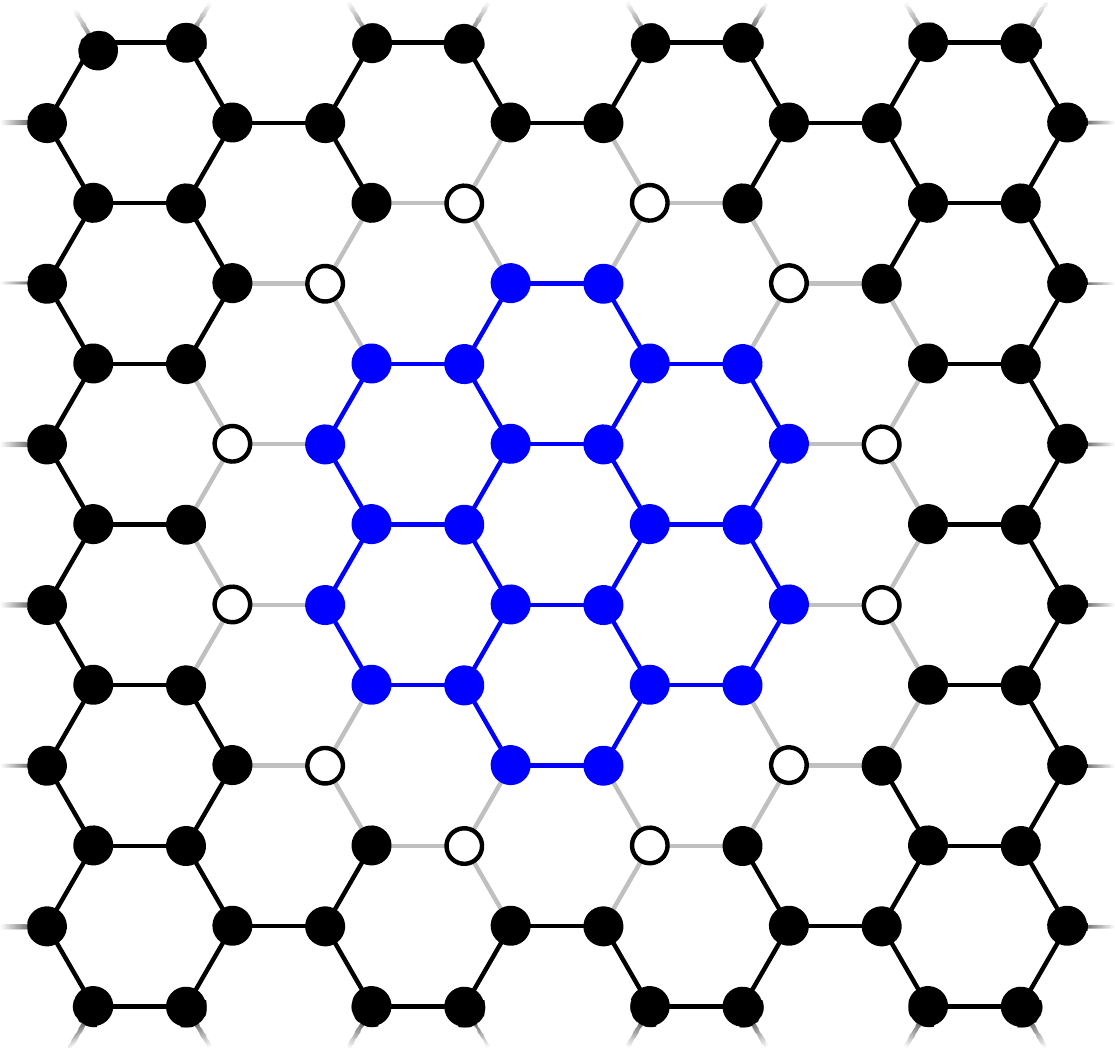}
\caption{Periodic hexagonal lattice. The white vertices are considered to be the boundary vertices for the subgraph of the blue (interior) vertices.}
\label{Periodic hexagonal lattice}
\end{figure}
 
\begin{figure}[hbtp]
\centering
\includegraphics[width=5cm]{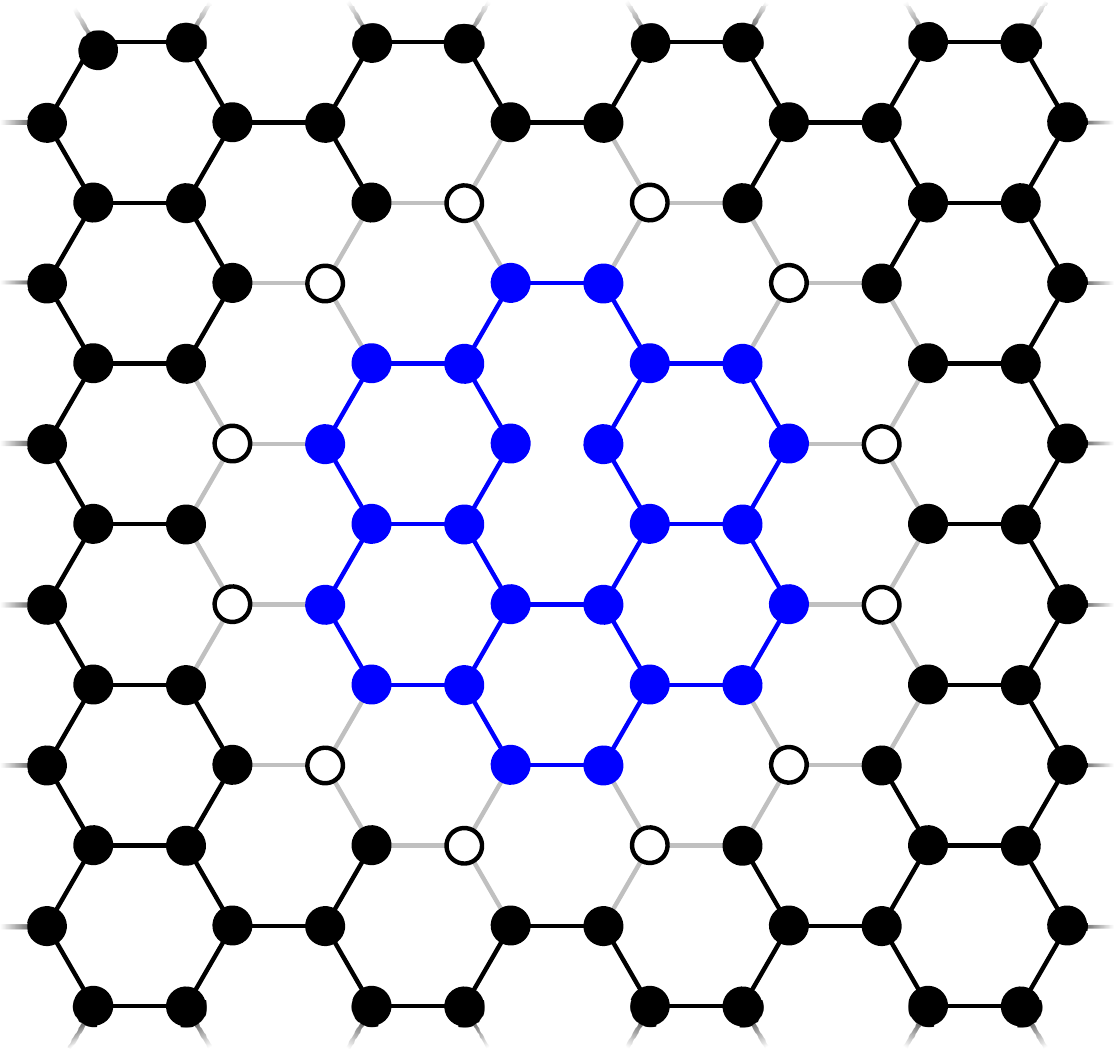}
\caption{A hexagonal lattice of Figure \ref{Periodic hexagonal lattice} with one blue edge removed. By Theorem \ref{mainscattering}, the exact structure of such graphs and the potential can be uniquely recovered from the scattering matrix.}
\label{Perturbed hexagonal lattice}
\end{figure}
\begin{figure}[hbtp]
\centering
\includegraphics[width=6cm]{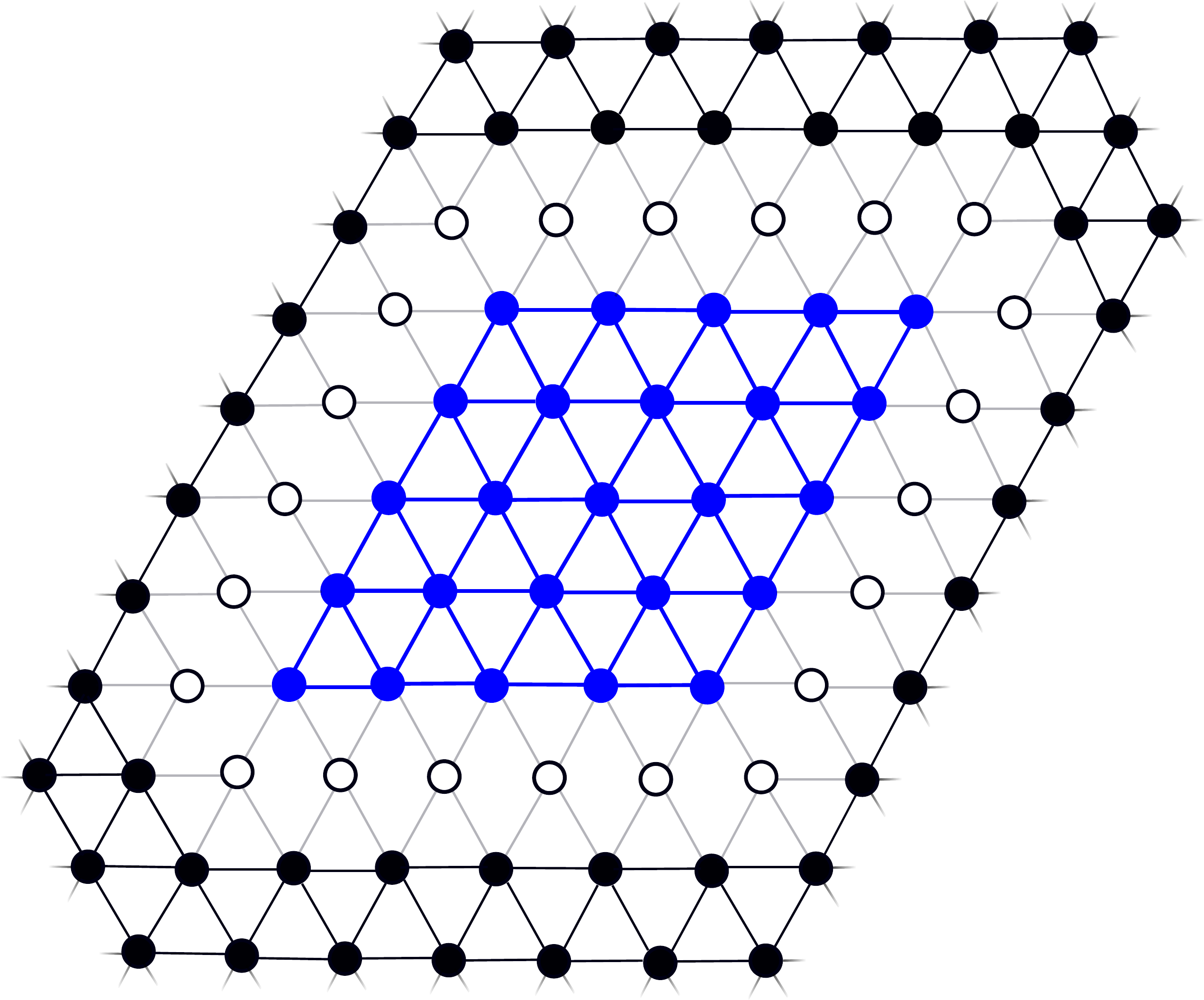}
\caption{A  triangular lattice satisfying (C-1)'.}
\label{Triangular lattice}
\end{figure}
\begin{figure}[hbtp]
\centering
\includegraphics[width=6cm]{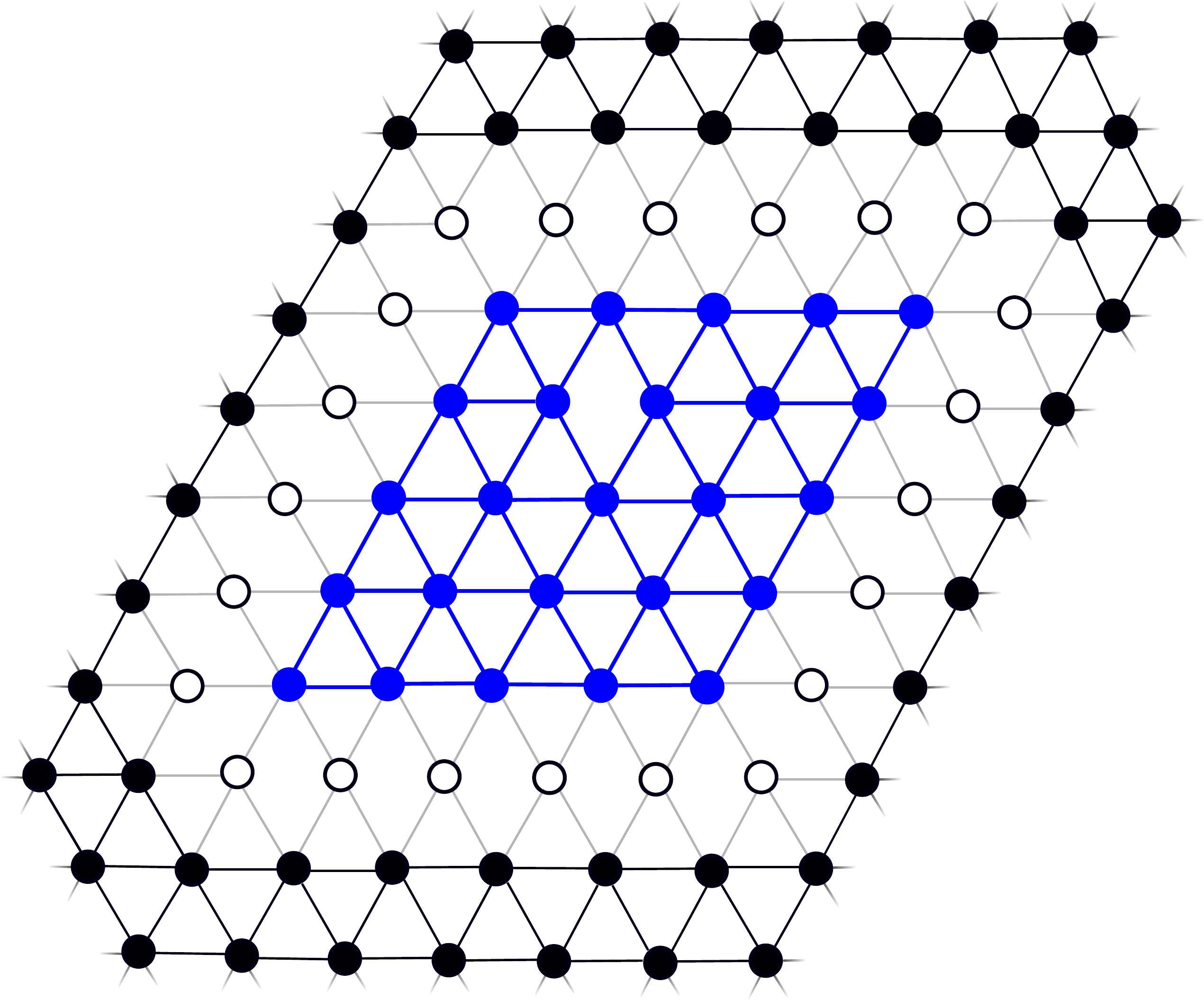}
\caption{A triangular lattice satisfying (C-1)' of Figure \ref{Triangular lattice} with one blue edge removed.  By Theorem \ref{mainscattering}, the exact structure of such graphs and the potential can be uniquely recovered from the scattering matrix.}
\label{Perturbed triangular lattice}
\end{figure}
\subsection{Periodic lattices and local perturbations}
\label{subsection-assumptions}
To begin with, we review a framework of the scattering theory on perturbed periodic lattices used in \cite{AIM16, AIM18}. 
 A periodic graph in $\mathbb{R}^d$ is a triple $\Gamma_0 = \{\mathcal L_0, \mathcal V_0, \mathcal E_0\}$, 
 where $\mathcal E_0$ is the edge set, and $\mathcal L_0$ is a lattice of rank $d$ in $\mathbb{R}^d$ with a basis ${\bf v}_j, j= 1,\cdots,d$, in other words
\begin{equation}
\mathcal L_0 = \big\{{\bf v}(n) : \, n \in
\mathbb{Z}^d\big\}, \quad
{\bf v}(n) = \sum_{j=1}^dn_j{\bf v}_j, \quad n =(n_1,\cdots,n_d) \in \mathbb{Z}^d.
\label{Definev(n)}
\end{equation}
The vertex set $\mathcal V_0$ is defined by
\begin{equation}
\mathcal V_0 = \bigcup_{j=1}^s \big(p_j + \mathcal L_0\big),
\label{S5VertexsetV0}
\end{equation}
where $p_j$, $j = 1,\cdots,s$, are points in $\mathbb{R}^d$ satisfying $p_i - p_j \not\in \mathcal L_0 \  {\rm if}\ i\neq j$. 
We assume that the degree of vertices are equal for all vertices $v \in \mathcal V_0$ and denote it by ${\rm deg}_{\mathcal V_0}$.
From (\ref{S5VertexsetV0}), we know that any function $\widehat f$ on $\mathcal V_0$ can be written as
$\widehat f(n) = (\widehat f_1(n),\cdots,\widehat f_s(n)), \  n \in \mathbb{Z}^d$, where $\widehat f_j(n)$ is  a function on $p_j + \mathcal L_0$. Hence the associated Hilbert space is $\ell^2(\mathcal V_0) = \ell^2(\mathbb{Z}^d)^s$, and it is unitarily equivalent to $L^2({\mathbb{T}}^d)^s$, where ${\mathbb{T}}^d$ is the flat torus $\mathbb{R}^d/(2\pi\mathbb{Z})^d$, by means of the  discrete Fourier transformation
\begin{equation}
(\mathcal U_{\mathcal V}\widehat f)(x) = \sqrt{\rm deg_{\mathcal V_0}}\,(2\pi)^{-d/2}\sum_{n\in{\mathbb Z}^d}\widehat f(n)\,e^{in\cdot x}, 
\quad {x \in {\mathbb T}^d},
\label{DefineMathcalUmathcalV}
\end{equation}
The Laplacian $\widehat \Delta_{\Gamma_0}$ on the lattice $\Gamma_0$ is defined by
\begin{equation}\label{Laplaciandef-3}
  \big(\widehat{\Delta}_{\Gamma_0} u\big)(v) = \frac{1}{{\rm deg}_{\mathcal{V}_0}}\sum_{w\in \mathcal{V}_0,e_{vw}\in \mathcal{E}_0} u(w),\quad v\in \mathcal{V}_0,
\end{equation}
{where, $e_{vw}$ denotes an edge $\in \mathcal E_0$ with end points $v, w \in \mathcal V_0$,}
 and we will use the symbol $\widehat H_0 = - \widehat\Delta_{\Gamma_0}$.

On the torus ${\mathbb{T}}^d = \mathbb{R}^d/(2\pi \mathbb{Z})^d$, the Floquet image of the Laplacian $\widehat{H}_0$ is an $s\times s$ matrix operator $H_0(x)$, where $x\in \mathbb{T}^d$ is the quasimomentum variable. We denote the matrix by $H_0$; its entries are trigonometric functions. Let $\lambda_1(x) \leq \cdots \leq \lambda_s(x)$ be the eigenvalues of $H_0(x)$. We put
\begin{align*}
p(x,\lambda) \, :=\,
\det(H_0(x) - \lambda),& \quad M_{\lambda} \,:=\, \{x \in {\mathbb{T}}^d : \,
p(x,\lambda) = 0\},
\\[.3em]
M_{\lambda,j} \,:=\, \{x \in {\mathbb{T}}^d :\, \lambda_j(x) = \lambda\},& \quad  M_{\lambda}^\mathbb{C} \,:=\, \{z \in \mathbb{C}^d/(2\pi\mathbb{Z})^d :\, p(z,\lambda) = 0\},
\\[.3em]
M_{\lambda,reg}^\mathbb{C} \,:=\,
\{z \in M_{\lambda}^\mathbb{C} :\, \nabla_zp(z,\lambda) \neq 0\},& \quad
M_{\lambda,sng}^\mathbb{C} \,:=\, \{z \in M_{\lambda}^\mathbb{C} :\, \nabla_zp(z,\lambda) = 0\}.
\end{align*}
In the spirit of \S\ref{s: The D-N maps}, we define
\begin{equation} \label{def-boundaryGamma0}
\partial_{\Gamma_0} \Omega :=\big\{v\in \mathcal{V}_0  \setminus \Omega\,\big|\, e_{vw} \in \mathcal{E}_0 \textrm{ for some }w\in \Omega \big\}.
\end{equation}
 
We impose the following assumptions on the periodic lattice $\Gamma_0$.

\medskip
\noindent
\textbf{(D-1)}
 {\it There exists a subset $\mathcal{T}_1\subset \sigma(H_0)$ such that for $\lambda\in \sigma(H_0)\setminus\mathcal{T}_1$, $M_{\lambda,sng}^{\mathbb{C}}$ is discrete, and each connected component of $M_{\lambda,reg}^{\mathbf{C}}$ intersects with $\mathbb{T}^d$, the intersection being a $(d-1)$-dimensional real analytic submanifold of $\mathbb{T}^d$.}

\smallskip
\noindent
\textbf{(D-2)} {\it There exists a finite set $\mathcal{T}_0\subset \sigma(H_0)$ such that}
$$
M_{\lambda,i}\cap M_{\lambda,j}=\emptyset\quad \textrm{ if }\:i\neq j \; \textrm{ and }\: \lambda\in \sigma(H_0)\setminus\mathcal{T}_0.
$$
\noindent
\textbf{(D-3)} {\it \  $\nabla_x p(x,\lambda)\neq 0\,$ holds on $M_{\lambda}$ for $\lambda\in \sigma(H_0)\setminus\mathcal{T}_0$.}

\smallskip
\noindent
\textbf{(D-4)}\, {\it The last assumption consists of two requirements: (i) On the unperturbed lattice $\Gamma_0$, there exist finite connected subsets $\{\Omega_k\}_{k=1}^{\infty}$ of $\mathcal{V}_0$ such that $\Omega_k\subset \Omega_{k+1},\,\mathcal{V}_0=\cup_{k=1}^{\infty}\Omega_k$, and the triple $(\Omega_k,\partial_{\Gamma_0}\Omega_k,\mathcal{E}_0)$ satisfies assumptions (C-1), (C-2) for all $k$, and (ii) the unique continuation from infinity holds on $\Omega_k^{ext} := \mathcal V_0\setminus \Omega_k$ for all $k$.}

\medskip
Assumption (D-4) requires a little explanation. For a subset $U \subset \mathcal V_0$ satisfying 
$\sharp (\mathcal V_0\setminus U) < \infty$, by the unique continuation from infinity on $U$, we mean the following claim. If $\hat u$ satisfies $(- \widehat{\Delta}_{\Gamma_0}- \lambda)\hat u = 0$ on $U$ for some $\lambda$ and $\hat u = 0$ near infinity, then  $\hat u$ vanishes on  whole $U$. Namely, if $\hat u$ satisfies $(- \widehat{\Delta}_{\Gamma_0}- \lambda)\hat u = 0$ on $U$ and $\hat u = 0$ on $|v| > R$ for some $R > 0$, then $\hat u = 0$ on $U$. 

On the other hand, the unique continuation from the boundary in the finite domain $\Omega_i$ follows from the first part of (D-4). 
Namely, if $(- \widehat{\Delta}_{\Gamma_0} - \lambda)\hat u = 0$ in $\Omega_i$ 
and $\hat u = \partial_{\nu}\hat u = 0$ on $\partial_{\Gamma_0}\Omega_i$, then $\hat u = 0$ in $\Omega_i$. This claim also holds for $- \widehat{ \Delta}_{\Gamma_0} + q(v)$ with any potential $q$, see Lemma 3.5 in \cite{BILL} or Lemma 2.4 in \cite {BILL2}.}

In particular, part (i) of (D-4)  implies the unique continuation property on $\mathcal V_0$ from infinity. 

\begin{lemma} 
\label{uc-infinity}
If part (i) of (D-4) is satisfied for $\Gamma_0$, then the unique continuation from infinity
holds for the unperturbed  equation $(-\widehat{\Delta}_{\Gamma_0} -\lambda)\hat u=0$ on $\Gamma_0$.
\end{lemma}

\begin{proof}
If a solution $\hat u$ is finitely supported in $\mathcal V_0$, we can find $\Omega_k$ such that ${\rm supp}\,( \hat u)\subset \Omega_k$ by assumption (i) of (D-4). Then $\hat u$ vanishes outside $\Omega_k$ on the unperturbed lattice $\Gamma_0$ for some $k$. By definition (\ref{Laplaciandef-3}), we know for any $z\in \partial_{\Gamma_0}\Omega_k$,
\begin{align*}
\sum_{x\sim z, x\in \Omega_i}\big( \hat{u}(x)- \hat{u}(z)\big)
&= \sum_{x\sim z, x\in \mathcal{V}_0} \big(\hat{u}(x)- \hat{u}(z)\big) \\
&= {\rm deg}_{\mathcal{E}_0}(z)\, \widehat{\Delta}_{\Gamma_0} \hat{u}(z) =-{\rm deg}_{\mathcal{E}_0}(z)\,\lambda \hat{u}(z)=0.
\end{align*}
This indicates that $\hat u$ is a solution of the equation (\ref{Schro}) on $\big(\Omega_k,\partial_{\Gamma_0}\Omega_k,\mathcal{E}_0 \big)$ satisfying simultaneously the Dirichlet and Neumann boundary conditions. Hence $\hat u$ vanishes everywhere by Lemma 2.4 in \cite {BILL2}, provided that the subgraph $\big(\Omega_k,\partial_{\Gamma_0}\Omega_k,\mathcal{E}_0 \big)$ satisfies the assumptions (C-1) and (C-2).
\end{proof}

The assumption  (D-2) implies that the eigenvalues $\lambda_j(x)$ are simple for $\lambda \not\in \mathcal T_0$. 
For $\lambda \not\in \mathcal T_1$, (D-1) guarantees the Rellich type theorem (cf. Theorems 5.1 and 5.7 in \cite{AIM16}). Therefore, (D-1) and (D-4) yield the non-existence of embedded eigenvalues for $H_0(x)$ and its perturbation for the  energy $\lambda \not\in 
\mathcal T_0 \cup\mathcal T_1$.

For the square, triangular, hexagonal, Kagome, and diamond lattices, as well as for subdivisions of square lattices, 
the subset $\mathcal T_1$ is finite. On the other hand, for the ladder and `layered' graphite lattices, $\mathcal T_1$ fills closed intervals, cf.~ \S 5 in \cite{AIM16}.

By virtue of Proposition 1.10 in \cite {BILL}, our result applies to several standard types of lattices and their perturbations. As for examples illustrating (i) of (D-4), see Example \ref{ExampleTwoPointsCond} of the present  paper. The unique continuation from infinity on 
$\Omega_i^{ext}$ is seen to be satisfied for e.g.  the square, hexagonal, triangular lattices by directly examining the figures.

{Referring to the papers \cite{AIM16}, \cite{AIM18}, we note that their authors employed four assumptions, (A-1)--(A-4), of which the first three coincided with (D-1)--(D-3) above. The fourth assumption there, (A-4), follows from part (i) of (D-4) by Lemma \ref{uc-infinity}.

\smallskip
Now let us consider an infinite connected graph $\Gamma = \{\mathcal V, \mathcal E\}$, which is a local (meaning compactly supported) perturbation of the periodic lattice $\Gamma_0 = \{\mathcal L_0, \mathcal V_0, \mathcal E_0\}$ satisfying the assumptions (D-1)--(D-4) above. We assume that the lattice $\Gamma_0$ is perturbed only in a finite subset $\Omega\subset \mathcal V_0$ and the potential function is also supported in $\Omega$. Later we will further assume (C-1) and (C-2) 
for the perturbed part in $\Omega$.
Lemma \ref{uc-infinity} then holds also for the perturbed system by the same proof, see Lemma \ref{uc-infinity-perturbed}.

Let $\{G,E_{pert}\}$ be a finite connected graph which is a perturbation of the subgraph $\big\{\Omega,\:\big\{e_{vw}\in \mathcal{E}_0: v, w\in \Omega \big\}\big\}$ of $\Gamma_0$. Without loss of generality, we may assume $\Omega$ is chosen sufficiently large so that the perturbation does not remove the vertices (of $\Omega$) which are connected to the subgraph boundary $\partial_{\Gamma_0} \Omega$. We add an unperturbed layer of edges to $E_{pert}$ defining
\begin{equation} \label{def-GammaE}
E \,:=\, E_{pert}\cup \big\{ e_{vw}\in \mathcal{E}_0 \,\big|\, v\in \Omega, w\in \partial_{\Gamma_0}\Omega \big\}.
\end{equation}
Then the weighted graph
\begin{equation}\label{def-G-lattice}
\mathbb{G}_{\Gamma}:=\{G,\partial_{\Gamma_0}\Omega ,E,\mu,g\},
\end{equation}
where $\mu= \{\mu_v\, ; \, v \in G\}$, $g = \{g_{vw}\, ; \, v, w \in G, \ v \sim w\}$ are the vertex weight and edge weight, fits into our setting for finite graphs studied in \cite{BILL}. For the scattering problem in this section, we set 
$$\partial G=\partial_{\Gamma_0}\Omega.$$ 
Observe that the edges connecting $\partial G$ and $G$ are known, and that by construction there are no edges between vertices in $\partial G$.

In particular, we can simply choose the perturbed vertex set $\Omega$ to be $\Omega_k$ for some $k$ as assumed in part (i) of (D-4).
We define the following sets:
\begin{align}\label{intext}
&\mathcal{V}_{int} \,:=\, G\cup\partial G,\quad {\mathcal V^\circ_{int}} \,:=\, G,\quad \partial \mathcal{V}_{int} \,:=\, \partial G;\nonumber \\
&\mathcal{V}_{ext} \,:=\, \mathcal{V} \setminus G,\quad {\mathcal V^{\circ}_{ext}} \,:=\, (\mathcal{V} \setminus G) \setminus \partial G,\quad \partial \mathcal{V}_{ext}  \,:=\,  \partial G.
\end{align}
Then the unique continuation from infinity holds on $\mathcal{V}_{ext}$ due to part (ii) of (D-4). Hence $\mathcal{V}_{int}$ and $\mathcal{V}_{ext}$ satisfy assumptions (B-1)--(B-3) imposed in \cite{AIM18}, and consequently, the Hilbert space $\ell^2(\mathcal V)$ admits an orthogonal decomposition
$$
\ell^2(\mathcal V) = \ell^2({\mathcal V^\circ_{ext})}\oplus \, \ell^2({\mathcal V_{int}}).
$$
Denote by $\widehat P_{ext}$ the orthogonal projection:
\begin{equation}
\widehat P_{ext} : \ell^2(\mathcal V) \to \ell^2({\mathcal V^{\circ}_{ext}}).
\label{Pextdefine}
\nonumber
\end{equation}
The Laplacian ${\widehat \Delta}_{\Gamma}$ on the graph $\Gamma$ is defined in analogy with (\ref{Laplaciandef-3}), replacing $\mathcal{V}_0,\mathcal{E}_0$ by $\mathcal{V},\mathcal{E}$.
Adding then a bounded self-adjoint perturbation of $\widehat V$, which is assumed to vanish on $\mathcal V_{ext}$, we consider Hamiltonian $\widehat H$ of the form
\begin{equation}
\widehat H = - \widehat\Delta_{\Gamma} + \widehat V 
{: \ell^2(\mathcal V) \to \ell^2(\mathcal V)}.
\nonumber
\end{equation}

\smallskip
Note that in the forward scattering problem, following the arguments of \cite{AIM16} and those from \S 2--\S 5 of \cite{AIM18}, one can allow arbitrary structure modification on the finite part of the graph.


\subsection{Spectral representation and the S-matrix}
\label{susubsectionspectralrepreSmatrix}
Let us keep reviewing the needed results from \cite{AIM16} and \cite{AIM18}. In general, scattering is a time-dependent phenomenon, and the S-matrix is defined through the wave operators. However, it has the stationary counterpart which we employ here. Let us recall how it looks for a Schr\"odinger operator in ${\mathbb R}^n$. We introduce a Banach space $\mathcal B({\mathbb R}^n)^{\ast}$ consisting of $L^2_{loc}({\mathbb R}^n)$ functions $f(x)$ such that
\begin{equation}
\|f\|^2_{{\mathcal B}({\mathbb R}^n)^\ast} := \sup_{R>1}\frac{1}{R}\int_{|x|<R}|f(x)|^2dx < \infty,
\end{equation}
which is the dual space of the Banach space $\mathcal B({\mathbb R}^n)$ defined as follows,
\begin{equation}
\|f\|_{\mathcal B({\mathbb R}^n)} = \sum_{j=0}^{\infty}R_j
\left(\int_{\Omega_j}|f(x)|^2dx\right)^{1/2} < \infty,
\label{DefinemathcalBRn}
\end{equation}
where $R_j = 2^j$ and $\Omega_j = \{x \in {\mathbb R}^d\, ; \, R_{j-1} \leq |x| < R_j\}$; for $j=0$ we put $R_{-1}:=0$. These spaces give rise to a rigged structure of $L^2({\mathbb R}^n)$, namely
$$
\mathcal B \subset L^2({\mathbb R}^n) \subset \mathcal B^{\ast}
$$
with continuous inclusions. Given $u, v \in \mathcal B({\mathbb R}^n)^{\ast}$, we define
\begin{equation} \label{simeq}
u \simeq v \Longleftrightarrow \lim_{R\to\infty}\frac{1}{R}\int_{|x|<R} |u(x) - v(x)|^2dx = 0.
\end{equation}
We consider the Helmholtz equation
\begin{equation}
(- \Delta + V(x) - \lambda)u = 0 \quad {\rm in} \quad {\mathbb R}^n,
\label{C4SRnSchrodingerEq}
\end{equation}
where $\lambda > 0$ and $V(x)$ is a real function decaying sufficiently rapidly at infinity. Then, for any $\phi^{in} \in L^2(S^{n-1})$, there exist a unique $u \in \mathcal B({\mathbb R}^n)^{\ast}$ satisfying (\ref{C4SRnSchrodingerEq}) and $\phi^{out} \in L^2(S^{n-1})$ such that
\begin{equation}
u \simeq \frac{e^{i\sqrt{\lambda}r}}{r^{(n-1)/2}}\phi^{out} - \frac {e^{-i\sqrt{\lambda}r}}{r^{(n-1)/2}}\phi^{in}.
\end{equation}
The operator
$$
S(\lambda) : L^2(S^{n-1}) \ni \phi^{in} \to \phi^{out} \in L^2(S^{n-1})
$$
is unitary and can be identified, up to a unitary operator, with the on-shell S-matrix obtained by the direct-integral decomposition of the scattering operator defined in the time-dependent theory.

As for scattering on perturbed periodic lattices, in some cases one can argue in the same way as above, e.g., when a square lattice is concerned \cite{IsoMo15}. However, to deal with general lattices, it is more convenient to pass the problem on the torus by the discrete Fourier transform and to observe the singularities of solutions to the Helmholtz equation.

On the torus ${\mathbb T}^d$, the counterpart of the above space $\mathcal B({\mathbb R}^n)^{\ast}$ is defined  as follows. Let $\phi$ be a distribution on ${\mathbb  T}^d$. Multiplying it by a smooth cut-off function, passing to the Fourier transform in the appropriate local chart, and denoting the resulting function by $\widetilde\phi$, we define $\mathcal B({\mathbb T}^d)^{\ast}$ to be the set of distributions such that
\begin{equation}
\sup_{R>1}\frac{1}{R}\int_{|\xi|<R}|\widetilde\phi(\xi)|^2d\xi < \infty\,;
\end{equation}
for two distributions $\phi, \psi$ on ${\mathbb T}^d$, $\phi \simeq \psi$ means
\begin{equation}
\frac{1}{R}\int_{|\xi| < R}|\widetilde\phi(\xi) - \widetilde\psi(\xi)|^2d\xi \to 0 \quad {\rm as}  \quad {R \to \infty.}
\label{Equationsimmomentum}
\end{equation}
{We also define the space $\mathcal B({\mathbb T}^d)$  similarly to (\ref{DefinemathcalBRn}). See \S 4 of \cite{AIM16} and \S 2.4 of \cite{AIM18}.} 

Assume that the unperturbed periodic lattice $\Gamma_0$ satisfies the above assumptions (D-1)--(D-4). The spectral representation of $H_0$ is nothing but the diagonalization of $H_0(x)$. Let  $P_j(x)$ be the eigenprojection associated with the eigenvalue $\lambda_j(x)$. Let $I_j = \{\lambda_j(x)\, ; \, x \in \mathbb{T}^d\}\setminus \mathcal T_0$, and
\begin{equation}
M_{\lambda,j} = \left\{
\begin{split}
 & \{x \in \mathbb{T}^d\, ; \, \lambda_j(x) = \lambda\}, \quad \lambda \in I_j, \\
& {\emptyset}, \quad \lambda \not\in I_j.
\end{split}
\right.
\end{equation}
For
$\lambda \in \sigma(H_0)\setminus\mathcal T_0$, we have $M_{\lambda,i} \cap M_{\lambda,j} = \emptyset$ if $i \neq j$, hence each of them is a $C^{\infty}$-submanifold of ${\mathbb T}^d$.
We define the Hilbert spaces ${\bf h}_{\lambda,j}$ equipped with the inner product
$$
(\psi,\phi)_{L^2(M_{\lambda,j})} = \int_{M_{\lambda,j}}P_j(x)\psi(x)\cdot \overline{\phi(x)}\frac{dM_{\lambda,j}}{|\nabla\lambda_j(x)|},
$$
and put
\begin{equation}
{\bf h_{\lambda}} = {\bf h}_{\lambda,1} \oplus \cdots \oplus {\bf h}_{\lambda,s}.
\end{equation}
For $f \in \mathcal B({\mathbb T}^d)$, we define
\begin{equation}
\mathcal F_{0,j}(\lambda)f =  P_j(x)f(x)\big|_{M_{\lambda,j}}
\end{equation}
and
\begin{equation}
\mathcal F_0(\lambda)f = (\mathcal F_{0,1}(\lambda)f,\dots,\mathcal F_{0,s}(\lambda)f);
\end{equation}
in the spirit of the above orthogonal sum, we often write the right-hand side as $\sum_{j=1}^s\mathcal F_{0,j}(\lambda)f$.
Then the operators
\begin{equation}
\mathcal F_0(\lambda) \in {\bf B}(\mathcal B({\mathbb T}^d)\, ; \, {\bf h}_{\lambda})
\end{equation}
provide us with a spectral representation (or a generalized Fourier transformation) associated with $H_0$. It is related to the resolvent of $H_0$ in the following way,
\begin{equation}
(H_0 - \lambda \mp i0)^{-1}f \simeq \sum_{j=1}^s\frac{\mathcal F_{0,j}(\lambda)f}{\lambda_j(x) - \lambda \mp i0}, \quad f \in \mathcal B({\mathbb T}^d),
\end{equation}
{where the relation $\simeq$ is defined by \eqref{Equationsimmomentum}}. This shows that the generalized Fourier transform can be associated with the singular part of the resolvent of $H_0$ on the torus, which in turn describes the behavior at infinity of the resolvent of $\widehat H_0$ in the lattice space. Compared with the case of ${\mathbb R}^n$, the lattice and the torus here can be matched off against the position space and the momentum space, respectively.

The same fact  holds for the perturbed operator $\widehat H = - \widehat\Delta_{\Gamma} + \widehat V$ on $\ell^2(\mathcal V)$. One can easily check that $\sigma_e(\widehat H) = \sigma(\widehat H_0) = \sigma(H_0)$, and furthermore, that $\sigma_p(\widehat H)\cap \sigma_e(\widehat H)$ is discrete in $\sigma_e(H_0)\setminus \mathcal T_0$ with possible accumulation points in $\mathcal T_0$ only \cite[Lemma~7.5]{AIM16}. In the following we consider $\lambda \in \sigma(H_0)\setminus(\mathcal T_0\cup\sigma_p(\widehat H))$. Define $\mathcal B = \mathcal B(\mathcal V)$ and ${\mathcal B^{\ast}} = \mathcal B(\mathcal V)^{\ast}$ as direct sums,
\begin{equation}
\mathcal B(\mathcal V) = \mathcal B(\mathcal V_{ext}) \oplus \ell^2({{\mathcal V}^{\circ}_{int}}), \quad \mathcal B(\mathcal V)^{\ast} = \mathcal B(\mathcal V_{ext})^{\ast} \oplus \ell^2({{\mathcal V}^{\circ}_{int}}),
\end{equation}
where the spaces $\mathcal B(\mathcal V_{ext})$ and $\mathcal B(\mathcal V_{ext})^{\ast}$ are defined on the torus in the way described above\footnote{More explicitly, the norm of ${\mathcal B(\mathcal V_0)}^{\ast}$ is defined by
$
\|\widehat u\|^2_{{\mathcal B(\mathcal V_0)}^{\ast}}= \sup_{R>1}
\frac{1}{R}\sum_{|n|<R}|\widehat u(n)|^2,
$
while in the case of $\mathcal V_{ext}$, the sum ranges over vertices of the set $\mathcal V_{ext}$ only.}.
Denoting $\widehat R(z) := (\widehat H - z)^{-1}$ and assuming $\lambda \in \sigma_e(\widehat H)\setminus (\mathcal T_0 \cup \sigma_p(\widehat H))$, we have
\begin{equation}
\widehat R(\lambda \pm i0) \in \bf{B}(\mathcal B\, ; \, \mathcal B^{\ast}).
\end{equation}
The generalized Fourier transformation $\mathcal F_{\pm}(\lambda)$ associated with $\widehat H$ is given by\footnote{To get \eqref{genFourier}, we employ the resolvent equation, cf. the argument preceding Theorem~\ref{EigenfunctionExpansionWholespace} in \S \ref{SubsecSpectralRepresentation2}, in particular, the formula (\ref{DefineFjlambdaperturbed}).}
\begin{equation} \label{genFourier}
\mathcal F_{\pm}(\lambda) = \mathcal F_0(\lambda)\,\mathcal U_{\mathcal V}\widehat Q_1(\lambda \pm i0)\,\mathcal U_{\mathcal V}^{\ast},
\end{equation}
where
\begin{equation}
\widehat Q_1(z) = \widehat P_{ext} + \widehat K_1\widehat R(z),
\quad \widehat K_1 = \widehat H_0\widehat P_{ext} - \widehat P_{ext}\widehat H.
\end{equation}
It is related to the resolvent in the following way, see Theorems~7.7 and 7.15 in \cite{AIM16}:
\begin{theorem}
Let $\lambda \in \sigma_e(\widehat H)\setminus \big(\mathcal T_0 \cup \sigma_p(\widehat H)\big)$. For $f \in \mathcal B$ we have the relation
\begin{equation}
\mathcal U_{\mathcal V}\widehat P_{ext}\widehat R(\lambda \pm i0)f \simeq \sum_{j=1}^s\frac{\mathcal F_{\pm, j}(\lambda)f}{\lambda_j(x) - \lambda \mp i0}.
\end{equation}
\end{theorem}

As in the case of ${\mathbb R}^n$, the S-matrix is defined by means of the Helmholtz equation.

\begin{theorem}
\label{TheoremHelmhotzeqmathcalV}
(1) For any solution $\widehat u \in \widehat{\mathcal B}$ to the equation
\begin{equation}
(\widehat H - \lambda)\widehat u = 0,
\end{equation}
there exist unique pair of vectors $\phi^{in}, \phi^{out} \in {\bf h}_{\lambda}$ such that
\begin{equation}
\mathcal U_{\mathcal V}\widehat P_{ext}\widehat u \simeq \sum_{j=1}^s \frac{1}{2\pi i}\left(\frac{\phi^{out}_j}{\lambda_j(x) - \lambda - i0} - \frac{\phi^{in}_j}{\lambda_j(x) - \lambda + i0}\right).
\label{Expansionofuforvertexcase}
\end{equation}
Moreover, the operator $S(\lambda) \in \bf{B}({\bf h}_{\lambda}\, ; \, {\bf h}_{\lambda})$ defined by
\begin{equation}
S(\lambda) = 1 - 2\pi iA(\lambda),
\label{S4.2Smatrix=1-2piAlambda}
\end{equation}
where
\begin{equation}
A(\lambda) := \mathcal F_+(\lambda)\,\mathcal U_{\mathcal V}\widehat K_2\,\mathcal U_{\mathcal V}^{\ast}\mathcal F_0(\lambda)^{\ast}, \quad \widehat K_2 := \widehat H\widehat P_{ext} - \widehat P_{ext}\widehat H_0,
\label{S4.2DefineA(lambda)}
\end{equation}
is unitary on ${\bf h}_{\lambda}$, and satisfies
\begin{equation} \label{onshell}
\phi^{out} = S(\lambda)\phi^{in}.
\end{equation}
 (2)
For any $\phi^{in} \in {\bf h}_{\lambda}$, there is a unique $\widehat u \in \widehat{\mathcal B}$ and $\phi^{out} \in {\bf h}_{\lambda}$ such that
\begin{equation}
(\widehat H - \lambda)\widehat u = 0,
\end{equation}
and relations \eqref{Expansionofuforvertexcase}, \eqref{onshell} are satisfied.
\end{theorem}
The operator $S(\lambda)$ is the S-matrix for our perturbed lattice, in the physics literature usually referred to as the on-shell S-matrix.

\subsection{The S-matrix and Dirichlet-to-Neumann map}
Now we consider eigenvalue equations separately on $\mathcal{V}_{ext}$ and $\mathcal{V}_{int}$, assuming that (B-1)--(B-3) of \cite{AIM18} are satisfied. Suppose that there is no perturbation outside $\mathcal{V}_{int}$ and that the potential is also supported in $\mathcal{V}_{int}$ only. For $\lambda \in \sigma_e(\widehat H_{ext})\setminus (\mathcal{T}_0\cup\mathcal{T}_1)$, there exists a unique solution $\widehat u^{(\pm)}_{ext} \in  \widehat{\mathcal B}^{\ast}$ to the following equation,
\begin{equation}
\left\{
\begin{split}
(- \widehat\Delta_{\Gamma_0} - \lambda)\widehat u^{(\pm)}_{ext} = 0 \quad
{\rm in} \quad {\mathcal V^{\circ}_{ext}}, \\
\widehat u^{(\pm)}_{ext} = \widehat f \quad {\rm on} \quad \partial\mathcal V_{ext},
\end{split}
\right.
\nonumber
\end{equation}
satisfying the radiation condition\footnote{We speak here of the discrete analogue of the usual radiation condition, see Section 2.6 in \cite{AIM18}.} (outgoing for $\widehat u^{(+)}_{ext}$ and incoming for $\widehat u^{(-)}_{ext}$). We define the {\it exterior D-N map} $\Lambda^{(\pm)}_{ext}(\lambda)$ by
\begin{equation}
\Lambda^{(\pm)}_{ext}(\lambda)\widehat f = - \partial_{\nu}^{\mathcal V_{ext}}\widehat u^{(\pm)}_{ext}\Big|_{\partial{\mathcal V}_{ext}},
\nonumber
\end{equation}
where the normal derivative of a function $u$ at $z\in \partial \mathcal{V}_{ext}$ in $\mathcal V^{\circ}_{ext}$ is defined by
\begin{equation}\label{Neumanndef-3}
\begin{split}
& \big(\partial_{\nu}^{\mathcal V_{ext}} u\big)(z) := -\frac{1}{{\rm deg}^{ext}_{\mathcal{E}}(z)}\sum_{x\in \mathcal V^{\circ}_{ext},\{x,z\}\in \mathcal{E}} u(x)\, ,\\
&{\rm deg}_{\mathcal{E}}^{ext}(z) :=\sharp\{x\in\mathcal V^{\circ}_{ext}\,:\,\{x,z\}\in \mathcal{E}\}.
\end{split}
\end{equation}
On the other hand, for $\lambda \not\in \sigma(\widehat H_{int})$, where $\widehat H_{int}$ is $-  \widehat\Delta_{\Gamma} + \widehat V$ in $\mathcal V_{int}$ with Dirichlet boundary condition,
there exists a unique solution $\widehat u_{int}$ to the following equation,
\begin{equation}
\left\{
\begin{split}
(- \widehat\Delta_{\Gamma} + \widehat V -\lambda)\widehat u_{int} = 0 \quad {\rm in} \quad \mathcal V^{\circ}_{int} ,\\
\widehat u_{int} = \widehat f \quad {\rm on} \quad \partial{\mathcal V}_{int}.
\end{split}
\right.
\nonumber
\end{equation}
The {\it interior D-N map} $\Lambda_{int}(\lambda)$ is defined by
\begin{equation}
\Lambda_{int}(\lambda)\widehat f =  \partial_{\nu}^{\mathcal V_{int}}\widehat u_{int}\Big|_{\partial{\mathcal V}_{int}},
\nonumber
\end{equation}
where the normal derivative at $\partial \mathcal{V}_{int}$ in $\mathcal V^{\circ}_{int}$ is defined in the analogous way, replacing all the exterior sets in (\ref{Neumanndef-3}) with the respective interior ones.

We denote 
\begin{equation}
\Sigma = \partial \mathcal V_{int} = \partial \mathcal V_{ext}
\label{DefineSigma}
\end{equation}
and define the operator
\begin{equation}
B^{(\pm)}_{\Sigma}(\lambda) := \mathcal M_{int}\Lambda_{int}(\lambda) -
\mathcal M_{ext}\Lambda^{(\pm)}_{ext}(\lambda) - \widehat S_{\Sigma} - \lambda\chi_{\Sigma},
\label{DefineBpmSigmalambda}
\end{equation}
where the operators $\mathcal M_{int}$, $\mathcal M_{ext}$, $\widehat S_{\Sigma}$, $\chi_{\Sigma}$ in \eqref{DefineBpmSigmalambda} contain only information referring to $\Sigma$; for their definitions we refer to relations (3.30)-(3.33) in \cite{AIM18}.

Next, letting 
$$
\widehat u^{(\pm)} = \chi_{{\mathcal V}_{int}^{\circ}}\widehat u_{int} + \chi_{{\mathcal V}_{ext}^{\circ}}\widehat u_{ext} + \chi_{\Sigma}\widehat f
$$
where, as above,  the operators $\chi_{{\mathcal V}_{int}^{\circ}}$ and $\chi_{{\mathcal V}_{ext}^{\circ}}$ contain only information referring to ${\mathcal V}_{int}^{\circ}$ and ${\mathcal V}_{ext}^{\circ}$, 
we define another operator, $\widehat I^{(\pm)}(\lambda): \ell^2(\Sigma) \to {\bf h}_{\lambda}$, by
{(see (4.7) in \cite{AIM18}), 
\begin{equation}
\widehat I^{(\pm)}(\lambda)\widehat f := \mathcal F_0(\lambda)\mathcal U(\widehat H_0 - \lambda)\widehat P_{ext}\widehat u^{(\pm)};
\label{Ipmlabda=mathcalFH-lambdauext}
\nonumber
\end{equation}
the right-hand side of this relation shows that the action of $I^{(\pm)}(\lambda)$ depends neither on $\mathcal V_{int}$ nor on $\widehat V$, in other words, it is independent of the perturbation.

Here, an important role is played by a Rellich-type result, Theorem 5.1 in \cite{AIM16}, and the following unique continuation property: if a solution of $(-\widehat{\Delta}_{\Gamma_0}-\lambda)\hat{u}=0$ on $\mathcal{V}_0$ vanishes except for a finite number of vertices for $\lambda\in \mathbb{C}$, then this solution vanishes identically on $\mathcal{V}_0$. This is what was assumed as (A-4) in \cite{AIM16,AIM18}. 
The said Rellich-type theorem, together with the unique continuation property in the exterior domain $\mathcal V_{ext}$ (which follows from the assumption (D-4)), implies the following claim, cf.~Lemma 4.3 in \cite{AIM18}.


\begin{lemma}
Let $\lambda \in \sigma_e(\widehat H)\setminus \big(\mathcal T_0\cup\mathcal T_1\cup\sigma_p(\widehat{H})\cup\sigma(\widehat H_{int})\big)$.
Then \\
\noindent (1) the map $\ \widehat I^{(\pm)}(\lambda) : \ell^2(\Sigma) \to {\bf h}_{\lambda}$ is injective, \\
(2) its adjoint $\ \widehat I^{(\pm)}(\lambda)^{\ast}  : {\bf h}_{\lambda} \to \ell^2(\Sigma)$ is surjective.
\end{lemma}
The scattering amplitude $A(\lambda)$ is defined by (\ref{S4.2DefineA(lambda)}).
In a similar way one can define the scattering amplitude in the exterior domain which we denote as $A_{ext}(\lambda)$. These scattering amplitudes satisfy the following relation, cf.~Theorem 4.5 in \cite{AIM18}.


\begin{theorem}
\label{ThSmatrixDNmapEquiv}
Let $\lambda \in \sigma_e(\widehat H)\setminus(\mathcal T_0\cup\mathcal T_1\cup\sigma_p(\widehat{H})\cup\sigma(\widehat H_{int}))$.
Then we have
\begin{equation}
A_{ext}(\lambda) - A(\lambda) = \widehat I^{(+)}(\lambda)\big(B^{(+)}_{\Sigma}(\lambda)\big)^{-1}
\widehat I^{(-)}(\lambda)^{\ast}.
\label{Aext-A=I+MI-}
\end{equation}
\end{theorem}
By assumption, the exterior domain is free of perturbations, therefore $\Lambda_{ext}^{(\pm)}(\lambda)$ and $A_{ext}(\lambda)$ are known. By virtue of (\ref{S4.2Smatrix=1-2piAlambda}), (\ref{DefineBpmSigmalambda}) and Theorem \ref{ThSmatrixDNmapEquiv}, the S-matrix  $S(\lambda)$ and the D-N map $\Lambda_{int}(\lambda)$ determine  each other on some interval in the spectrum, and the same is true for the N-D map. Since the S-matrix, the D-N map and the N-D map are all complex analytic, this mutual determination extends from the said interval to the whole spectrum. Thus we arrive at the following claim.
\begin{theorem}
 For any $\lambda \in \sigma_e(\widehat H)\setminus(\mathcal T_0\cup \mathcal T_1\cup\sigma_p(\widehat{H})\cup \sigma(\widehat H_{int}))$, the S-matrix $S(\lambda)$ and the D-N map $\Lambda_{int}(\lambda)$ determine each other.
\end{theorem}
Let us remark that the definition of the normal derivative  used in \cite{AIM18} differs from the present one given by (\ref{S3partialnudefine}), adopted from \cite{BILL}, by a constant only. Hence the corresponding Neumann-to-Dirichlet maps determine each other.

Note further that the formula (\ref{Aext-A=I+MI-}) is a discrete analogue of the one derived by Isakov and Nachman in \cite{IsaNach} for the Schr{\"o}dinger operator in $\mathbb{R}^n$. For the discrete problem, it provides us with a constructive route from the S-matrix to the corresponding D-N map.

\smallskip

\subsection{The inverse scattering problem}
\label{section-scatteringproof}

The aim of this subsection is to show that the graph structure and the potential can be uniquely recovered from the knowledge of the scattering matrix at all energies for the discrete Schr\"odinger operator.

First of all, let us recall the definition of the Neumann-to-Dirichlet (N-D) map for a  finite weighted graph with boundary, $\mathbb{G}=\{G,\partial G,E,\mu,g\}$.
Let $q$ be a real-valued potential function on $G$, and denote by $\{\lambda_k\}_{k=1}^{N}$ the Neumann eigenvalues, with the multiplicity taken into account, of the discrete Schr\"odinger operator $-\Delta_G+q$, where $N=\sharp G$. We consider the following equation:
\begin{equation}\label{Schro}
    \begin{cases}
      (-\Delta_G+q-\lambda) u(x)=0, & x\in G,\; \lambda\in \mathbb{C},\\[.3em]
      \partial_{\nu} u \big|_{\partial G}=f,
    \end{cases}
\end{equation}
where the Neumann boundary value $\partial_{\nu}u$ was defined in \eqref{S3partialnudefine}.
For $\lambda\notin \{\lambda_k\}_{k=1}^N$, denote by $u_{\lambda}^f$ the unique solution of the equation (\ref{Schro}) with the Neumann boundary value equal to $f$. The Neumann-to-Dirichlet map $\Lambda_{\lambda}$ (at a fixed energy $\lambda$) for the equation (\ref{Schro}) is defined as $\Lambda_{\lambda}: f\mapsto u_{\lambda}^f \big|_{\partial G}$.

\begin{lemma}\label{ND}
Let $\mathbb{G}$ be a finite connected weighted graph with boundary satisfying the assumptions (C-1) and (C-2) in \S \ref{SectionFinitegraphInverseproblem}. Suppose the weights\footnote{We abuse the notation here writing $g|_{\partial G \times G}$ to indicate the weights of the edges connecting the boundary vertices with the interior vertices.} $\mu|_{\partial G}$, $g|_{\partial G\times G}$ are given. Then knowing the Neumann-to-Dirichlet map at all energies for the equation (\ref{Schro}) on $\mathbb{G}$ is equivalent to the knowledge the Neumann boundary spectral data for the discrete Schr\"odinger operator on $\mathbb G$.
\end{lemma}

\begin{proof}
The proof for the manifold case can be found in \cite{KKLM} or Section 4.1 of \cite{KKL}. The proof in our case, for finite graphs, is simpler. Let $\{\phi_k\}_{k=1}^N$ be a family of orthonormalized Neumann eigenfunctions of the discrete Schr\"odinger operator corresponding to eigenvalues $\{\lambda_k\}$. 
Recall from \cite{BILL} that the $L^2(G)$-inner product is defined by
$$\langle u_1,u_2 \rangle_{L^2(G)}= \sum_{x\in G} \mu_x u_1(x) u_2(x).$$
By Green's formula \cite[Lemma 2.1]{BILL}, we infer that
\begin{align*}
\big\langle (-\Delta_G+q)u_{\lambda}^f,\phi_k \big\rangle_{L^2(G)} &=\big\langle u_{\lambda}^f,(-\Delta_G+q)\phi_k \big\rangle_{L^2(G)} - \sum_{z\in \partial G} \mu_z\phi_k(z) (\partial_{\nu}u_{\lambda}^f)(z) \\
&= \lambda_k \langle u_{\lambda}^f,\phi_k \rangle_{L^2(G)} - \sum_{z\in \partial G} \mu_z\phi_k(z) f(z),
\end{align*}
which yields
$$(\lambda-\lambda_k) \langle u_{\lambda}^f,\phi_k \rangle_{L^2(G)}=- \sum_{z\in \partial G} \mu_z\phi_k(z)f(z) .$$
Now take an arbitrary real-valued function $w^f$ on $G\cup\partial G$ satisfying $\partial_{\nu} w^f|_{\partial G}=f$. Then the difference
$u_{\lambda}^f-w^f$ lies in the domain of the Neumann graph Laplacian and we have
\begin{align}
u_{\lambda}^f-w^f& = \sum_{k=1}^{N}\langle u_{\lambda}^f-w^f , \phi_k \rangle_{L^2(G)} \phi_k \nonumber \\
&= -\sum_{k=1}^{N}\frac{1}{\lambda-\lambda_k}\Big(\sum_{z\in \partial G} \mu_z\phi_k(z)f(z) \Big) \phi_k - \sum_{k=1}^{N}\langle w^f , \phi_k \rangle_{L^2(G)} \phi_k. \label{eq-ND}
\end{align}
This shows that $\Lambda_{\lambda}$ is a meromorphic operator-valued function of $\lambda$ with simple poles at $\lambda=\lambda_k$ only, and this in turn means that $\{\Lambda_{\lambda}\}$ determines the set of eigenvalues $\{\lambda_k\}$.
Moreover, the residue of $\Lambda_{\lambda}$ at $\lambda=\lambda_k$ is known as a finite-dimensional linear operator. In particular, since $\mu|_{\partial G}$ is known, the data $\{\Lambda_{\lambda}\}$ determine
$$
Q_k(z_1,z_2)=\sum_{l\in L_k}\phi_l(z_1)\phi_l(z_2), \quad \forall\, z_1,z_2\in \partial G,
$$
where $L_k=\{p_k+1,\cdots,p_k+\sharp L_k\}$, $p_k\in \mathbb{N}$,
denotes the set of integers $l$ satisfying $\lambda_l=\lambda_k$. This function $Q_k(\cdot,\cdot)$ can be viewed as an $m\times m$ matrix $Q_k$ defined by $(Q_k)_{ij}=Q_k(z_i,z_j)$, where $m=\sharp \partial G$, or in the matrix form
$$Q_k=\big(\phi_{p_k+1},\cdots,\phi_{p_k+\sharp L_k}\big)_{m\times \sharp L_k} \, \big(\phi_{p_k+1},\cdots,\phi_{p_k+\sharp L_k} \big)_{m\times \sharp L_k}^T\, .$$
By Lemma 2.4 of \cite{BILL2}, the eigenfunctions $\{\phi_l|_{\partial G}\}_{l\in L_k}$ are linearly independent on $\partial G$, hence the rank of $Q_k$ is simply $\sharp L_k$.

When the eigenvalue $\lambda_k$ is simple, the matrix $Q_k$ determines $\phi_k|_{\partial G}$ up to the sign. In general, since $Q_k$ is symmetric and positive semi-definite, it can be decomposed into $Q_k=BB^T$, where $B$ is an $m\times \sharp L_k$ matrix of rank $\sharp L_k$. Moreover, the decomposition is unique up to an $\sharp L_k\times \sharp L_k$ orthogonal matrix. Thus we take the column vectors of $B$, and they are the boundary values of orthonormalized eigenfunctions found by applying the orthogonal matrix to $\{\phi_l\}_{l\in L_k}$. This shows $Q_k$ determines the boundary values of the orthonormalized eigenfunctions (referring to the choice of $\{\phi_k\}_{k=1}^N$ we made).

To check the converse: the Neumann boundary spectral data determine the N-D map in accordance with the formula \eqref{eq-ND}. We choose $w^f$ such that $w^f|_G=0$ and $\partial_{\nu} w^f|_{\partial G}=f$ so that the last term in \eqref{eq-ND} vanishes. Since $g|_{\partial G\times G}$ is known, thus $w^f|_{\partial G}$ is uniquely determined by $f$, and consequently, the N-D map can be determined from the Neumann boundary spectral data.
\end{proof}

Without loss of generality, we assume that 
the perturbed vertex set $\Omega=\Omega_{k_0}$ for some $k_0$ as assumed in part (i) of (D-4), cf.~ \S \ref{subsection-assumptions}.
With our choice (\ref{intext}) of the domains,
Theorem \ref{ThSmatrixDNmapEquiv} and Lemma \ref{uc-infinity} yield the following statement.

\begin{cor}\label{Smatrix}
Let $\Gamma_0$ be an infinite periodic lattice satisfying assumptions (D-1)--{(D-4)}. 
Let $q$ be a finitely supported potential on $\Gamma$, and $\mathbb{G}_{\Gamma}$ be the perturbed finite subgraph given by \eqref{def-G-lattice}. Then the knowledge of the scattering matrix of the discrete Schr\"odinger operator on $\Gamma$ {at an arbitrarily fixed energy} determines the Neumann-to-Dirichlet map of the equation (\ref{Schro}) on $\mathbb{G}_{\Gamma}$ with $\mu={\rm deg}_E,\, g\equiv 1$ for {the same energy}.
\end{cor}

Now we impose the following assumption on the locally perturbed lattice $\Gamma$.

\medskip
\noindent
\textbf{(E-1)} With the perturbed vertex set $\Omega=\Omega_{k_0}$ for some $k_0$ as in part (i) of (D-4), the perturbed finite subgraph $\mathbb{G}_{\Gamma}$ given by \eqref{def-G-lattice} is connected and satisfies (C-1), (C-2).
\medskip

The assumption (E-1), together with part (ii) of (D-4), implies the unique continuation from infinity for the perturbed system.

\begin{lemma} \label{uc-infinity-perturbed}
Assume (E-1) and part (ii) of (D-4) are satisfied. 
Then the unique continuation from infinity
holds for the perturbed equation $(-\widehat{\Delta}_{\Gamma} -\lambda)\hat u=0$ on $\Gamma$.
\end{lemma}

\begin{proof}
By assumption (E-1), the system is unperturbed outside of $\Omega_{k_0}$. If $\hat u$ vanishes near infinity, then $\hat u$ vanishes on $\mathcal{V}_0 \setminus \Omega_{k_0}$ due to part (ii) of (D-4). Then the lemma follows from the same argument as Lemma \ref{uc-infinity}.
\end{proof}

Our main result of this section is stated as follows.
 
\begin{theorem}\label{mainscattering}
Consider a periodic lattice satisfying assumptions (D-1)--(D-4),  and suppose that $\Gamma$ is an infinite graph obtained by a local perturbation of this lattice. Let the potential $q$ be finitely supported on $\Gamma$, and $\mathbb{G}_{\Gamma}$ be the perturbed finite subgraph given by \eqref{def-G-lattice} with $\mu={\rm deg}_E,\, g\equiv 1$. Assume that $\mathbb{G}_{\Gamma}$ satisfies (E-1). Then $\mathbb{G}_{\Gamma}$ and $q$ can be uniquely recovered from the knowledge of the scattering matrix for the discrete Schr\"odinger operator on $\Gamma$ for all energies.
\end{theorem}
\begin{proof}
From our construction of $\mathbb{G}_{\Gamma}$ in \S \ref{subsection-assumptions}, the edges connecting $\partial G$ and $G$ are known, and hence the weight $\mu={\rm deg}_E$ on $\partial G$ is known. 
{The theorem then follows from Corollary \ref{Smatrix} and Theorem \ref{TheoremBILL}}.
\end{proof}

\begin{example}
\label{ExampleTwoPointsCond}
{\rm Finite square, hexagonal (see Figure \ref{Periodic hexagonal lattice}), triangular, graphite and square ladder lattices all satisfy the Two-Points Condition (C-2) with the set of boundary vertices being the domain boundary.
Moreover, any horizontal edges in these lattices can be removed and the obtained graphs still  satisfy the Two-Points Condition, see Figure \ref{Perturbed hexagonal lattice}; the term
 ``horizontal edges'' here refers to the edges in the non-gradient directions with respect to the function $h$ in Proposition 1.8 in \cite{BILL}.}
\end{example}


\section{Spectral theory for periodic quantum graph}
\label{SectionQunatumGraphSpectralTheory}

In this and the next sections, we study the spectral theory for the Schr{\"o}dinger operator on a quantum (metric) graph. Let $\Gamma_0 = \{\mathcal L_0, \mathcal V_0, \mathcal E_0\}$ be a periodic lattice introduced in \S \ref{InverseScatteringDiscrete}, 
and let assumptions (D-1)--(D-4) be imposed. As in  \S \ref{subsection-assumptions}, we consider a local perturbation $\Gamma = \{\mathcal V, \mathcal E\}$ of $\Gamma_0$. On each edge $e \in \mathcal E$, we are given a one-dimensional  Schr{\"o}dinger operator $h_{e} = -d^2/dz^2 + V_{e}(z)$ satisfying the \mbox{$\delta$-coupling} condition (\ref{GeneralizedKirchhoff}) together with the assumptions (M-1)--(M-5) in \S \ref{Metricgraphsection}.
We assume that $V_e(z)$ is equal to a fixed potential $V_0(z)$ except for a  finite number of edges $e$. For the sake of (mainly notational) simplicity, we further assume that $V_0(z) = 0$ {and $\ell_e = 1$ for all edges $e$}. The arguments below also works for the general case by replacing $\phi^{(0)}_{e0}(z,\lambda)$, $\phi^{(0)}_{e1}(z,\lambda)$ and $\sigma^{(0)}(h^{(0)})$ by those associated with $V_0(z)$. Let $\widehat H_{\mathcal E}$ be the resulting self-adjoint operator in $L^2(\mathcal E)$. In the unperturbed case, when $V_e = 0$ holds for each $e \in \mathcal E_0$ and $C_v/d_v$ is equal to a fixed constant $\kappa_{\mathcal V}$, that is,
\begin{equation}
\frac{C_v}{d_v} = \kappa_{\mathcal V}, \quad \forall v \in \mathcal V_0,
\end{equation}
the operator $\widehat H_{\mathcal E}$ shall be denoted by $\widehat H^{(0)}_{\mathcal E}$. In what follows, we call $\widehat H_{\mathcal E}$ the `edge' Schr{\"o}dinger operator, and $-\widehat\Delta_{\mathcal V,\lambda}$ the `vertex' Schr{\"o}dinger operator.

\subsection{Spectrum of $\widehat H_{\mathcal E}$}

Let us begin with the unperturbed operator $\widehat H_{\mathcal E}^{(0)}$. Amending all the symbols introduced in \S \ref{Metricgraphsection} with the superscript $(0)$, we have $\phi^{(0)}_{e0}(z,\lambda) = \frac{\sin (\sqrt{\lambda} z)}{\sqrt{\lambda}}$ and $\phi^{(0)}_{e1}(z,\lambda) = \frac{\sin (\sqrt{\lambda}(1- z))}{\sqrt{\lambda}}$, hence 
\begin{equation}
\big(\widehat\Delta_{\mathcal V,\lambda}^{(0)}\widehat u\big)(v) =
\frac{\sqrt{\lambda}}{\sin\sqrt{\lambda}}\frac{1}{d_v}\sum_{w\sim v}\widehat u(w) = \frac{\sqrt{\lambda}}{\sin\sqrt{\lambda}}\big(\widehat \Delta_{\mathcal V}\hat u\big)(v),
\label{hatDelta0Vlambda}
\end{equation}
with $\widehat \Delta_{\mathcal V}$ being the vertex Laplacian on  $\mathcal V_0$, and
\begin{equation}
\widehat Q^{(0)}_{\mathcal V,\lambda} = \frac{\sqrt{\lambda}}{\sin\sqrt{\lambda}}\cos\sqrt{\lambda} + \kappa_{\mathcal V}.
\label{hatQ0Vlambda}
\end{equation}
We put
\begin{equation}
E(\lambda ) = -\cos\sqrt{\lambda} - \kappa_{\mathcal V}\frac{\sin\sqrt{\lambda}}{\sqrt{\lambda}},
\label{DefineE(lambda)}
\end{equation}
and then the resolvent $R^{(0)}_{\mathcal E}(\lambda) = (H^{(0)}_{\mathcal E} - \lambda)^{-1}$ can be, in view of Lemma \ref{S2LemmaRElambdaformula}, rewritten as
\begin{equation}
R^{(0)}_{\mathcal E}(\lambda) = \big(\widehat T^{(0)}_{\mathcal V,\overline{\lambda}}\big)^{\ast}
\frac{\sin\sqrt{\lambda}}{\sqrt{\lambda}}\Big( - \widehat\Delta_{\mathcal V} - E(\lambda)\Big)^{-1}\widehat T^{(0)}_{\mathcal V,\overline{\lambda}} + r^{(0)}_{\mathcal E}(\lambda).
\label{resolventequationR(0)mathcalE}
\end{equation}
Furthermore, we put
\begin{equation}
\sigma^{(0)}(h^{(0)}) = \{(\pi j)^2\, ; \, j = 1, 2, \dots\},
\label{Definsigma0h0}
\end{equation}
\begin{equation}
\sigma^{(0)}(-\widehat\Delta_{\mathcal V}) =
\{\lambda \, ; \, E(\lambda )\in \sigma(- \widehat\Delta_{\mathcal V})\},
\label{Definesigam0deltaV}
\end{equation}
\begin{equation}
\sigma^{(0)}_{\mathcal T} = \{\lambda \in {\rm Int}\,(\sigma_e(\widehat H_{\mathcal E}^{(0)}));\:  E(\lambda) \in \mathcal T\}, 
\label{Definesigma0tau}
\end{equation}
where ${\rm Int}\, I$ for a subset $I \subset \mathbb{R}$ means the interior of $I$, and
\begin{equation}
\mathcal T = \mathcal T_0 \cup \mathcal T_1.
\label{Definetau0tau1sp0h0)}
\end{equation}

Relation (\ref{resolventequationR(0)mathcalE}) allows us to write the spectrum in the following way:
\begin{lemma}
$\ \sigma(\widehat H^{(0)}_{\mathcal E}) = \sigma^{(0)}(-\widehat\Delta_{\mathcal V}) \cup \sigma^{(0)}(h^{(0)}).$
\end{lemma}

For example, in the Kirchhoff coupling case, $\kappa_{\mathcal V} = 0$, we have $\sigma(\widehat H^{(0)}_{\mathcal E}) = [0,\infty)$ for square and hexagonal lattices. Note that $\sigma^{(0)}(h^{(0)})$ is the set of eigenvalues of infinite multiplicities embedded in $\sigma(\widehat H^{(0)}_{\mathcal E})$.

\subsection{Function spaces}

For an edge $e \in \mathcal E_0$ with the endpoints $v, w \in \mathcal V_0$, we define
\begin{equation} \label{midpoint}
|e_c| = \frac{1}{2}|v + w|,
\end{equation}
i.e. the distance of its midpoint from the origin, where for $x = (x_1,\dots,x_d)\in {\mathcal V_0} \subset {\mathbb  R}^d$ we denote $|x| = \sqrt{x_1^2 + \dots + x_d^2}$. It will serve as a radius-like variable allowing to define the needed function spaces. Recall that our graph $\Gamma = (\mathcal V, \mathcal E)$ is a local perturbation of a periodic lattice $\Gamma_0 = (\mathcal L_0, \mathcal V_0, \mathcal E_0)$, which means that $\Gamma$ and $\Gamma_0$ coincide in the exterior domain
\begin{equation}
\mathcal E_{ext,R} \ni e \Longleftrightarrow |e_c| > R,
\end{equation}
provided $R$ is chosen sufficiently large; without loss of generality we may suppose that $R > 1$. The interior domain
\begin{equation}
\mathcal E_{int,R} = \mathcal E\setminus \mathcal E_{ext,R}
\end{equation}
in which all the perturbations are located is finite and the `radius' plays no role there. Hence we keep the definition \eqref{midpoint} in the exterior domain, and for the interior domain $\mathcal E_{int,R}$ we put instead
\begin{equation}
|e_c| = 1 \quad {\rm if} \quad e \in \mathcal E_{int,R}.
\end{equation}
With this proviso we introduce the function spaces on ${\mathcal E}$: we put $r_j = 2^{j}$ and define
\begin{equation} \label{L2s}
\widehat L^{2,s}(\mathcal E) \ni \hat f \Longleftrightarrow
\sum_{e \in \mathcal E}|e_c|^{2s}\|\hat f_e\|^2_{L^2(e)} < \infty,
\end{equation}
\begin{equation} \label{Bspace}
\widehat{\mathcal B}(\mathcal E) \ni \widehat f \Longleftrightarrow
 \sum_{e \in \mathcal E}r_j^{1/2}\Big(
\sum_{r_{j-1} \leq |e_c| < r_j}\|\hat f_e\|_{L^2(e)}^2 \Big)^{1/2} < \infty,
\end{equation}
\begin{equation} \label{B*space}
\widehat{\mathcal B}^{\ast}(\mathcal E) \ni \hat f\Longleftrightarrow
\sup_{R>1}\frac{1}{R}\sum_{|e_c| < R}\|\hat f_e\|^2_{L^2(e)} < \infty,
\end{equation}
equipped with their obvious norms. As the notation suggests, $\widehat{\mathcal B}^{\ast}(\mathcal E)$ can be identified with the dual space of $\widehat{\mathcal B}(\mathcal E)$, and the following inclusions hold for $s > 1/2$:
\begin{equation}
\widehat L^{2,s}(\mathcal E)\subset \widehat{\mathcal B}(\mathcal E)\subset \widehat L^{2,1/2}(\mathcal E) \subset \widehat L^{2}(\mathcal E) \subset \widehat L^{2,-1/2}(\mathcal E)\subset \widehat{\mathcal B}^{\ast}(\mathcal E)\subset \widehat L^{2,-s}(\mathcal E),
\end{equation}
where $\widehat L^2(\mathcal E) = \widehat L^{2,0}(\mathcal E)$. Moreover, $\widehat{\mathcal B}^{\ast}_0(\mathcal E)$ is a closed subspace of $\widehat{\mathcal B}^{\ast}(\mathcal E)$
defined by
\begin{equation}
\widehat{\mathcal B}^{\ast}_0(\mathcal E) \ni \widehat f \Longleftrightarrow \lim_{R\to\infty}\frac{1}{R}\sum_{|e_c| < R}\|\hat f_e\|^2_{L^2(e)} = 0.
\end{equation}

Let us further note that for the `vertex' Laplacian, the spaces $\widehat L^{2,s}(\mathcal V)$, $\widehat{\mathcal B}(\mathcal V)$, $\widehat{\mathcal B}^{\ast}(\mathcal V)$, $\widehat{\mathcal B}^{\ast}_0(\mathcal V)$ are defined in the same way as above with the norms $\|\hat f_e\|_{L^2(e)}$ at the right-hand sides of \eqref{L2s}--\eqref{B*space} replaced by $|\hat f(v)|$. This is one more manifestation of the parallelism between the discrete graph and  the quantum graph. In the former, we consider ${\mathbb C}$-valued functions on the discrete set $\mathcal V$, while in the latter, we deal with $L^2((0,1))$-valued functions on the discrete set $\{e_c\, ; \, e \in \mathcal E_{ext,R}\}$. This correspondence is inherited, in particular, in the resolvent estimates.

\subsection{Rellich-type theorem}

\begin{theorem}
\label{QuantumgraphRellichtypeTheorem}
Let $\lambda \in \big({\rm Int}\, \sigma_e(\widehat H^{(0)}_{\mathcal E})\big)\setminus \sigma^{(0)}_{\mathcal T}$, 
and suppose that $\hat u \in \widehat{\mathcal B}^{\ast}_0(\mathcal E)$ satisfies $\widehat H^{(0)}_{\mathcal E}\hat u = \lambda \hat u$ and the $\delta$-coupling condition in $\mathcal E_{ext,R}$ for some $R>1$.
Then $\hat u = 0$ holds in $\mathcal E_{ext,R_1}$ for some $R_1 \geq R$.
\end{theorem}
\begin{proof}
Since $R$ is chosen large enough so that all the perturbations are inside of $\mathcal E_{int,R}$, on each edge $e \in \mathcal E_{ext,R}$, the solution $\hat u$ can be written as
$$
\hat u_e(z) = \hat u_e(1)\frac{\sin\sqrt{\lambda}z}{\sqrt{\lambda}} + \hat u_e(0)\frac{\sin\sqrt{\lambda}(1-z)}{\sqrt{\lambda}}.
$$
As the functions $\frac{\sin\sqrt{\lambda}(1-z)}{\sqrt{\lambda}}$ and $\frac{\sin\sqrt{\lambda}z}{\sqrt{\lambda}}$ are linearly independent for such a $\lambda$, there exists a constant $C_{\lambda} > 0$ such that
\begin{equation}
C_{\lambda}^{-1}\left(|\hat u_e(0)| + |\hat u_e(1)|\right) \leq \|\hat u_e\|_{L^2(e)} \leq C_{\lambda}\left(|\hat u_e(0)| + |\widehat u_e(1)|\right)
\label{S5u0u1andUL2eareequivalent}
\end{equation}
for all $e \in \mathcal E_{ext,R}$. We put $\hat w = \hat u\big|_{\mathcal V}$, then in view of Lemma \ref{S2LemmaDeltamathcalVmathcalV=Tlambda}, we have
$$
\big(- \widehat\Delta_{\mathcal V} - E(\lambda)\big)\hat w = 0, \quad {\rm on } \quad \mathcal V\cap \mathcal E_{ext,R}.
$$
Since $\hat u \in \widehat{\mathcal B}_0^{\ast}(\mathcal E)$ holds by assumption, the inequality (\ref{S5u0u1andUL2eareequivalent}) implies 
$\hat w \in \widehat{\mathcal B}^{\ast}_0(\mathcal V)$. By the Rellich-type theorem for vertex Schr{\"o}dinger operators \cite[Theorem~5.1]{AIM16}, we have $\widehat w(v) = 0$ for $|v| > R'$ with a sufficiently large $R'$. This proves the theorem.
\end{proof}
\begin{definition}
We say that the operator $\widehat H_{\mathcal E} - \lambda$ has the  unique continuation property if the following assertion holds: If $\hat u$ satisfies $(\widehat H_\mathcal{E} - \lambda)\hat u = 0$ on $\mathcal E$, and $\hat u = 0$ on $\mathcal E_{ext,R}$ for a positive $R$, then $\hat u = 0$ holds on $\mathcal E$.
\end{definition}

For the unperturbed system, by assumption (D-4) in \S \ref{InverseScatteringDiscrete} (essentially coinciding with (C-2) in \S \ref{SectionFinitegraphInverseproblem}), $\widehat H_{\mathcal E}^{(0)} - \lambda$ has the unique continuation property for all $\lambda$. 
{Adding a potential, it is also true for the unperturbed operator $\widehat H_{\mathcal E}$. }
\begin{lemma}
Under the assumptions (D-1)--{(D-4)}, we have
$$
\sigma_p(\widehat H^{(0)}_{\mathcal E})\cap \sigma_e(\widehat H^{(0)}_{\mathcal E}) \subset
 \sigma^{(0)}_{\mathcal T}.$$
\end{lemma}

\begin{proof}
Any eigenvector of $\widehat H_{\mathcal E}$ is in $\widehat{L}^2(\mathcal V) \subset {\widehat{\mathcal B}^{\ast}_0(\mathcal E)}$, and therefore it vanishes `at infinity' by Theorem \ref{QuantumgraphRellichtypeTheorem}. By the unique continuation property, it vanishes everywhere.
\end{proof}

As can be checked easily, the square  and hexagonal lattices satisfy the unique continuation property.

\subsection{Radiation condition}
For systems having ${\mathbb R}^d$ as the configuration space, the radiation condition is introduced either by observing the asymptotic behavior at infinity, or, what is equivalent, from the singularities of the Fourier image of solutions to the Schr{\"o}dinger equation. Dealing with lattice Schr{\"o}dinger operators, we adopt the latter approach.
\begin{definition}
Given a distribution $u \in \mathcal D'({\mathbb T}^d)$, its wave front set $WF^{\ast}(u)$ is defined as follows: a point $(x_0,\omega) \in {\mathbb R}^d\times S^{d-1}$ does not belong to  $WF^{\ast}(u)$  if and only if there exist $0 < \delta < 1$ and $\chi(x) \in C_0^{\infty}({\mathbb R}^d)$ such that $\chi(x_0) =1$ and
\begin{equation}
\mathop{\lim_{R\to\infty}}\frac{1}{R}\int_{|\xi|<R}|C_{\omega,\delta}(\xi)
(\widetilde{\chi u})(\xi)|^2d\xi = 0,
\end{equation}
where $\widetilde{\chi u}$ is the Fourier transform of $\chi u$ and $C_{\omega,\delta}(\xi)$  is the characteristic function of the cone $\{\xi \in {\mathbb R}^d\, ; \, \omega\cdot\xi > \delta|\xi|\}$.
\end{definition}
Let $\lambda_j(x)$, $j = 1,2,\dots,s$, be the eigenvalues of $H_0(x)$ and $P_j(x)$ the associated eigenprojections, and let $H_0$ be the operator of multiplication by $H_0(x)$ on $\big(L^2({\mathbb T}^d)\big)^s$. In \cite[Lemma~4.7]{AIM16}, it was proven that the operator
$$
\mathcal B(
{\mathbb T}^d) \ni f \to \frac{f(x)}{\lambda_j(x) - \rho \mp i0} \in \mathcal B^{\ast}({\mathbb T}^d)
$$
is bounded if $\rho \not\in \big({\rm Int}\, \sigma(H_0)\big)\setminus \mathcal T$. Furthermore in \cite[Theorem~6.1]{AIM16} it was shown that for any $f \in \mathcal B({\mathbb T}^d)$,  $1 \leq j \leq s$ and $\rho \in \sigma(H_0)\setminus \mathcal T$, it holds that

\medskip
\noindent
$(RC)_+$  : \ \  \
$\displaystyle{WF^{\ast}\big(\frac{P_jf}{\lambda_j(x) - \rho - i0}\big) \subset \{(x,\omega_x) \, ; \, x \in M_{\rho,j}\},}$

\medskip
\noindent
$(RC)_-$ : \ \  \
$\displaystyle{WF^{\ast}\big(\frac{P_jf}{\lambda_j(x) - \rho + i0}\big) \subset \{(x,-\omega_x) \, ; \, x \in M_{\rho,j}\}}$,

\medskip
\noindent
where $\omega_x \in S^{d-1}\cap T_x(M_{\lambda,j})^{\perp}$ and $\omega(x)\cdot \nabla \lambda_j(x) < 0$. Moreover, for any $f \in \mathcal B({\mathbb T}^d)$, the function $u = (H_0(x) - \lambda \mp i0)^{-1}f \in \mathcal B^{\ast}({\mathbb T}^d)$ is the unique solution to the equation $(H_0(x) - \rho)u = f$ satisfying $(RC)_+$ or $(RC)_-$, respectively. These claims also extend to the case with compactly supported perturbations.

We put
\begin{equation}
{\rm sgn}(\lambda) =
\left\{
\begin{split}
& 1 \quad {\rm for} \quad \lambda > 0, \quad \sin\sqrt{\lambda} > 0, \\
& -1 \quad {\rm for} \quad \lambda > 0, \quad \sin\sqrt{\lambda} < 0,
\end{split}
\right.
\end{equation}
and then we can write
\begin{equation}
\cos\sqrt{\lambda \pm i0} =  \cos\sqrt{\lambda} \mp i\,0 \,{\rm sgn}\,(\lambda), \quad \lambda > 0.
\label{Cossqrtlambda}
\end{equation}
 
We recall the discrete Fourier transform $\mathcal U_{\mathcal V}$  defined by (\ref{DefineMathcalUmathcalV}). Let $\widehat P_{ext,R}$ be the orthogonal projection : $L^2(\mathcal E) \to L^2(\mathcal E_{ext,R})$. Taking (\ref{Cossqrtlambda}) into account, we define the radiation condition as follows.

\begin{definition}
A solution $\hat u \in \widehat{\mathcal B}^{\ast}(\mathcal E)$ of the equation $(- \widehat\Delta_{\mathcal E} + V - \lambda)\hat u = \hat f$ is said to satisfy the {\it outgoing radiation condition} if either

\smallskip
\noindent
(i)  $\sin\sqrt{\lambda} > 0$, and $w = \mathcal U\widehat P_{ext,R}\hat u\big|_{\mathcal V}$ satisfies $(RC)_+$ with $\rho = E(\lambda)$,

\smallskip
or

\smallskip
\noindent
(ii)  $\sin\sqrt{\lambda} < 0$, and $u = \mathcal U\widehat P_{ext,R}\hat u\big|_{\mathcal V}$ satisfies $(RC)_-$  with $\rho = E(\lambda),$

\noindent
holds.
Similarly, we define the {\it  incoming radiation condition} with $(RC)_{\pm}$ replaced by $(RC)_{\mp}$. If $\hat u$ satisfies either the outgoing radiation condition or the incoming one, we simply say that $\hat u$ satisfies the radiation condition.
\end{definition}

In \cite{AIM16}, the radiation condition was also introduced for the vertex Laplacian, see Lemmata 4.8 and 6.2 there. Let $\hat f \in \mathcal{B} (\mathcal E)$. Given a solution $\hat u$ to the edge Schr{\"o}dinger equation $(- \widehat\Delta_{\mathcal E}+ {V} - \lambda)\hat u = \hat f$, denote by $\hat u\big|_{\mathcal V}$  its restriction to $\mathcal V$. Then $\hat u\big|_{\mathcal V}$ satisfies the vertex Schr{\"o}dinger equation
\begin{equation}
\Big( - \widehat\Delta_{\mathcal V} - E(\lambda)\Big)\hat u  = \hat g,
\end{equation}
where $\hat g \in \mathcal{B} (\mathcal V)$. Comparing these two definitions of the radiation condition, one can make the following claim:

\begin{lemma}
\label{RadCondequivlatticeedge} A solution
$\hat u$ of the edge Schr{\"o}dinger equation satisfies the radiation condition if and only if the solution $\hat u\big|_{\mathcal V}$ of the vertex Schr{\"o}dinger equation satisfies the radiation condition.
\end{lemma}
\begin{lemma}
\label{LemmaRadCondUnique}
Let $\lambda \in \big({\rm Int}\,\sigma_e(\widehat H_{\mathcal E})\big) \setminus \sigma^{(0)}_{\mathcal T}$. 
Then the solution $\hat u \in \widehat{\mathcal B}^{\ast}({\mathcal E})$ of the equation $(- \widehat\Delta_{\mathcal E} + V - \lambda)\widehat u = \widehat f$ satisfying the radiation condition is unique.
\end{lemma}
\begin{proof}
For the vertex Schr{\"o}dinger operator, such a result was proven in  Lemma 7.6 of \cite{AIM16}; in combination with Lemma \ref{RadCondequivlatticeedge}, it yields the claim for the edge Schr{\"o}dinger operator.
\end{proof}

\subsection{Limiting absorption principle}
\label{SubsectionLAP}

Let us first investigate the existence of the limits $( - \widehat\Delta_{\mathcal V,\lambda \pm i0} + \widehat Q_{\mathcal V,\lambda \pm i0})^{-1}$.
\begin{lemma}
\label{LemmahatDeltaVlambda-1}
If $E(\lambda) \in \sigma_e(- \widehat{\Delta}_{\mathcal V})\setminus \mathcal T$, there exists a limit
$$
( - \widehat\Delta_{\mathcal V,\lambda \pm i0} + \widehat Q_{\mathcal V,\lambda \pm i0})^{-1} \in {\mathbb B}(\mathcal B(\mathcal V);\mathcal B \big(\mathcal V)^{\ast} \big).
$$
\end{lemma}
\begin{proof}
We use the limiting absorption principle for the vertex Schr{\"o}dinger operator proved in \cite{AIM16}. Taking into account (\ref{hatDelta0Vlambda}) and (\ref{hatQ0Vlambda}), we define $\widehat W_{\mathcal V,\lambda}$ by
\begin{equation}
-\widehat\Delta_{\mathcal V,\lambda} + \widehat Q_{\mathcal V,\lambda} = \frac{\sqrt{\lambda}}{\sin\sqrt{\lambda}}\left(
-\widehat\Delta_{\mathcal V} -  E(\lambda) + \widehat W_{\mathcal V,\lambda} \right),
\end{equation}
where $ \widehat W_{\mathcal V,\lambda}$ is a self-adjoint, bounded, and compactly supported perturbation of $-\widehat\Delta_{\mathcal V} $. Then, regarding $E(\lambda)$ as the energy for $- \widehat\Delta_{\mathcal V}$, and arguing in the same way as in \cite{AIM16}, we can prove the existence of the limit
$$
( - \widehat\Delta_{\mathcal V}+ \widehat W_{\mathcal V,\lambda} - E(\lambda \pm i0))^{-1}.
$$
Using the identity
\begin{equation}
\begin{split}
 & - \widehat\Delta_{\mathcal V} + \widehat W_{\mathcal V,\lambda\pm i\epsilon} - E(\lambda \pm i\epsilon)\\
&  =  - \widehat\Delta_{\mathcal V}+ \widehat W_{\mathcal V,\lambda} - E(\lambda \pm i\epsilon) + (W_{\mathcal V,\lambda\pm i\epsilon} - W_{\mathcal V,\lambda}),
\end{split}
\end{equation}
together with the fact that $W_{\mathcal V,\lambda\pm i\epsilon} - W_{\mathcal V,\lambda} \to 0$ as $\epsilon \to 0$, we can construct the inverse of the right-hand side by the Neumann series. This proves the lemma.
\end{proof}

For $\lambda \not\in \cup_{e\in \mathcal E}\sigma_p(- (d/dz)^2_D + V_e(z))$, {where $-(d/dz)^2_D$ denotes $- (d/dz)^2$ in $L^2(e)$ with Dirichlet boundary condition}, the functions $\phi_{e0}(z,\lambda)$ and $\phi_{e1}(z,\lambda)$ are linearly independent, hence by (\ref{DefineTVlambdaasthatg=frac1dVphi/phig(0)}) there is a constant $C_{\lambda} > 0$ such that
\begin{equation}
C_{\lambda}^{-1}\left(|\hat g(e(0))| + |\hat g(e(1))|\right) \leq
\|(\widehat T_{\mathcal V,\lambda})^{\ast}\hat g_e\|_{L^2(e)} \leq C_{\lambda}\left(|\hat g(e(0))| + |\hat g(e(1))|\right)
\end{equation}
holds for all $e \in \mathcal E$. This implies
\begin{equation}
(\widehat T_{\mathcal V,\lambda})^{\ast} \in {\mathbb B} \big(\mathcal B(\mathcal V)^{\ast}\, ; \, \mathcal B^{\ast}(\mathcal E) \big),
\label{TVlambdaastboundedfromBVtoBE}
\end{equation}
and
\begin{equation}
\widehat T_{\mathcal V,\lambda} \in {\mathbb B}\big(\mathcal B(\mathcal E)\, ; \, \mathcal B(\mathcal V) \big).
\label{TVlambdaboundedfromBEtoBV}
\end{equation}

Combining Lemma \ref{S2LemmaRElambdaformula} with (\ref{TVlambdaastboundedfromBVtoBE}), (\ref{TVlambdaboundedfromBEtoBV}), we arrive at the following result.
\begin{theorem}
\label{LAPwholespace}
Let $I$ be a compact interval in $\big({\rm Int}\,\sigma_e(\widehat H_{\mathcal E})\big) \setminus  \sigma^{(0)}_{\mathcal T}$.  \\
\noindent
(1) There exists a constant $C > 0$ such that
\begin{equation}
\|( \widehat H_{\mathcal E} - \lambda \mp i\epsilon)^{-1}\|_{{\mathbb B}(\widehat{\mathcal B}(\mathcal E);\widehat{\mathcal B}^{\ast}(\mathcal E))} \leq C
\end{equation}
holds for any $\lambda \in I$ and $\epsilon > 0$. \\
\noindent
(2) For any $\lambda \in I$ and $s > 1/2$, there exist strong limits
\begin{equation}
{\mathop{\rm s-lim}_{\epsilon\downarrow 0}}( \widehat H_{\mathcal E} - \lambda \mp i\epsilon)^{-1} =: (\widehat H_{\mathcal E} - \lambda \mp i0)^{-1} \in
{\mathbb B}\big(\widehat L^{2,s}(\mathcal E);\widehat L^{2,-s}(\mathcal E)\big).
\end{equation}
\noindent
(3) For any $\widehat f \in \widehat L^{2,s}(\mathcal E)$, $(\widehat H_{\mathcal E} - \lambda \mp i0)^{-1}\widehat f$ is an $\widehat L^{2,-s}(\mathcal E)$-valued strongly continuous function of $\lambda \in I$.\\
\noindent
(4) For any $\widehat f, \widehat g \in \widehat{\mathcal B}(\mathcal E)$, there exist limits
\begin{equation}
{\mathop{\lim}_{\epsilon\downarrow 0}}\big((\widehat H_{\mathcal E} - \lambda \mp i\epsilon)^{-1}\widehat f,\widehat g\big) =: \big(( \widehat H_{\mathcal E} - \lambda \mp i0)^{-1}\widehat f,\widehat g\big),
\end{equation}
and $ \big((\widehat H_{\mathcal E} - \lambda \mp i0)^{-1}\widehat f,\widehat g\big)$ is a continuous function of $\lambda \in I$.

\noindent
(5) For any $\widehat f \in \widehat{\mathcal B}(\mathcal E)$, $( \widehat H_{\mathcal E} - \lambda - i0)^{-1}\widehat f$ satisfies the outgoing radiation condition, and $(\widehat H_{\mathcal E} - \lambda + i0)^{-1}\widehat f$ satisfies the incoming radiation condition.
\end{theorem}

\subsection{Spectral representation}
\label{SubsecSpectralRepresentation2}

As we have noted in the paragraph following eq.~\eqref{S5VertexsetV0}, there are unitary equivalences
$$
\ell^2(\mathcal V_0) \cong (\ell^2({\mathbb Z}^d))^s \cong (L^2({\mathbb T}^d))^s,
$$
by means of the decomposition (\ref{S5VertexsetV0}) and the discrete Fourier trandformation
(\ref{DefineMathcalUmathcalV}) with ${\rm deg}\,\mathcal E_0(x) = d_{\mathcal V_0}$. In the following, we  freely make use of the identification
\begin{equation}
\ell^2(\mathcal V_0) \ni \big(\hat f(v)\big)_{v \in \mathcal V_0} \longleftrightarrow
\hat f(n) = (\hat f_1(n),\dots,\hat f_s(n)) \in (\ell^2({\mathbb Z}^d))^s
\label{ell2(V0)unitaryell2(Zd)s}
\end{equation}
and we put\footnote{To be more precise, one should insert the operator of identification $J : \ell^2(\mathcal V_0) \to
(\ell^2({\mathbb Z}^d))^s$ defined by (\ref{ell2(V0)unitaryell2(Zd)s}) in front of $\widehat T_{\mathcal V,\lambda}^{(0)}$. We omit it, however, for the sake of simplicity.}
\begin{equation}
\Phi^{(0)}(\lambda) = \mathcal U_{\mathcal V}
{\widehat T_{\mathcal V,\lambda}^{(0)}},
\end{equation}
{ where  ${\widehat T_{\mathcal V,\lambda}^{(0)}}$ is the unperturbed $\widehat T_{\mathcal V,\lambda}$ defined by (\ref{DefineTVlambdaf(v)=1/dvsume(0)=vintIedz}).}
Let $P_{\mathcal V,j}(x)$ be the eigenprojection associated with the eigenvalue $\lambda_j(x)$ of $H_0(x)$, and denote
\begin{equation}
\begin{split}
D^{(0)}(\lambda \pm i0) &=  \frac{\sin{\sqrt{\lambda}}}{\sqrt{\lambda}}\,
\mathcal U_{\mathcal V}(- \widehat\Delta_{\mathcal V} - E(\lambda \pm i0))^{-1}\mathcal U_{\mathcal V}^{\ast}\\
&=
\frac{\sin{\sqrt{\lambda}}}{\sqrt{\lambda}}
\sum_{j=1}^s\frac{1}{\lambda_j(x) - E(\lambda \pm i0)}P_{\mathcal V,j}(x).
\end{split}
\end{equation}
By (\ref{resolventequationR(0)mathcalE}),
the following formula holds:
\begin{equation}
\begin{split}
\widehat R^{(0)}_{\mathcal E}(\lambda \pm i0) =
\Phi^{(0)}(\lambda)^{\ast}D^{(0)}(\lambda\pm i0)\Phi^{(0)}(\lambda) + r_{\mathcal E}^{(0)}(\lambda).
\end{split}
\label{LemmaR0lambdapi0mathcalE}
\end{equation}
To construct a spectral representation of $\widehat H^{(0)}_{\mathcal E}$, we put
\begin{equation}
M_{\mathcal E,\lambda,j} = \{x \in {\mathbb T}^d\, ; \, \lambda_j(x) - E(\lambda) = 0\},
\nonumber
\end{equation}
\begin{equation}
(\varphi,\psi)_{\lambda,j} = \int_{M_{\mathcal E,\lambda,j}}P_{\mathcal V,j}(x)\varphi(x)\cdot\overline{\psi(x)}\,dS_j,
\end{equation}
\begin{equation}
dS_j = \frac{|\sin \sqrt{\lambda}|}{\sqrt{\lambda}} \frac{dM_{\mathcal E,\lambda,j}}{|\nabla_x\lambda_j(x)|}.
\nonumber
\end{equation}
Combining (\ref{LemmaR0lambdapi0mathcalE}) with the formula
\begin{equation}
\begin{split}
& \big(\widehat R_{\mathcal V}^{(0)}(-\cos\sqrt{\lambda + i0}) \widehat f-
\widehat R_{\mathcal V}^{(0)}(-\cos\sqrt{\lambda - i0})\widehat f, \widehat g\big) \\
& = 2\pi i \sum_j\int_{M_{\mathcal E, \lambda, j}}
P_{\mathcal V,j}\widehat f \cdot \overline{P_{\mathcal V,j}\widehat g}\,
\frac{dM_{\mathcal E,\lambda,j}}{|\nabla \lambda_j(x)|},
\end{split}
\nonumber
\end{equation}
valid for $\lambda \in \big({\rm Int}\,\sigma_e(\widehat H^{(0)}_{\mathcal E})\big) \setminus \sigma^{(0)}_{\mathcal T}$,
for which we refer to eq.~(6.7) of \cite{AIM16}, we obtain the relation
\begin{equation}
\begin{split}
&  \frac{1}{2\pi i}\Big(\big(\widehat R^{(0)}_{\mathcal E}(\lambda + i0) - \widehat R^{(0)}_{\mathcal E}(\lambda - i0)\big)\widehat f, \widehat g\Big) = \sum_{j=1}^s
\big(P_{\mathcal V,j}\Phi^{(0)}(\lambda)\widehat f,P_{\mathcal V,j}\Phi^{(0)}(\lambda)\widehat g\big)_{\lambda,j}.
\end{split}
\label{Parseval0forR0lambda}
\end{equation}
Furthermore, we put
\begin{equation}
\widehat{\mathcal F}_j^{(0)}(\lambda)\widehat f = \left(P_{\mathcal V,j}\Phi^{(0)}(\lambda)\widehat f\right)\Big|_{M_{\mathcal E,\lambda,j}},
\label{S5DefineFj(0)lambda}
\end{equation}
in other words, the restriction to $M_{\mathcal E,\lambda,j}$ with the components
\begin{equation}
\widehat{\mathcal F}^{(0)}(\lambda) = \big(\widehat{\mathcal F}^{(0)}_{1}(\lambda),\cdots,\widehat{\mathcal F}^{(0)}_{s}(\lambda)\big),
\nonumber
\end{equation}
\begin{equation}
{\bf h}_{\lambda} = {\mathop\oplus_{j=1}^s}P_{\mathcal V,j}\Big|_{M_{\mathcal E,\lambda,j}}L^2\big(M_{\mathcal E,\lambda,j};dS_j\big),
\end{equation}
\begin{equation}
{\mathbb H} = L^2\big((0,\infty),{\bf h}_{\lambda};d\lambda\big).
\nonumber
\end{equation}
Then, by virtue of (\ref{Parseval0forR0lambda}) we can write
\begin{equation}
\frac{1}{2\pi i}\Big(\big(\widehat R^{(0)}_{\mathcal E}(\lambda + i0)-\widehat R^{(0)}_{\mathcal E}(\lambda - i0)\big)\widehat f,\widehat g\Big)
= (\widehat{\mathcal F}^{(0)}(\lambda)\widehat f,\widehat{\mathcal F}^{(0)}(\lambda)\widehat g)_{{\bf h}_{\lambda}}.
\nonumber
\end{equation}
Let $E^{(0)}(\lambda)$ be the spectral measure for $\widehat H^{(0)}_{\mathcal E}$. Integrating the last equality and using Stone's formula, we get
\begin{equation}
(E^{(0)}(I)\widehat f,\widehat g) = \int_I(\widehat{\mathcal F}^{(0)}(\lambda)\widehat f,\widehat{\mathcal F}^{(0)}(\lambda)\widehat g)_{{\bf h}_{\lambda}}d\lambda,
\nonumber
\end{equation}
for any interval $I \subset \big({\rm Int}\,\sigma_e(\widehat H^{(0)}_{\mathcal E})\big) \setminus \sigma^{(0)}_{\mathcal T}$.
Hence $\widehat{\mathcal F}^{(0)}$ extends uniquely to an isometry from the subspace\footnote{For a self-adjoint operator $A$, $\mathcal H_{ac}(A)$ denotes conventionally its absolutely continuous subspace, while $\mathcal H_{p}(A)$ is the closure of the linear hull of eigenvectors of $A$.}
$\mathcal H_{ac}(\widehat H^{(0)}_{\mathcal E})$ to ${\mathbb H}$. Moreover, we define
\begin{equation}
\widehat{\mathcal F}^{(0)} = 0, \quad {\rm on} \quad {\mathcal H}_p(\widehat H^{(0)}_{\mathcal E}).
\nonumber
\end{equation}
As one can see from (\ref{S5DefineFj(0)lambda}), to obtain $\widehat{\mathcal F}^{(0)}(\lambda)$ one has in fact to diagonalize the matrix $H_0(x)$.

The spectral representation  for $\widehat H_{\mathcal E}$ is constructed by the perturbation method well known in the stationary scattering theory. For the case of perturbation by a potential, we make use of the resolvent equation
\begin{equation}
\widehat R_{\mathcal E}(\lambda \pm i0) = \widehat R_{\mathcal E}^{(0)}(\lambda \pm i0)\big(1 -
V_{\mathcal E}\widehat R_{\mathcal E}(\lambda \pm i0)\big).
\label{Resolventequation1}
\end{equation}
Then, defining $\widehat{\mathcal F}^{(\pm)}(\lambda)$ by
\begin{equation}
\widehat{\mathcal F}^{(\pm)}(\lambda) = \widehat{\mathcal F}^{(0)}(\lambda)\big(1 - V_{\mathcal E}\widehat R_{\mathcal E}(\lambda \pm i0)\big) \in
{\mathbb B}(\widehat{\mathcal B}(\mathcal E)\, ; \,{\bf h}_{\lambda}),
\label{DefineFpmlambdainmathcalE}
\end{equation}
and using the resolvent equation \cite[Lemma~7.8]{AIM16}, we have
\begin{equation}
\frac{1}{2\pi i}\big(\big(\widehat R_{\mathcal E}(\lambda + i0)-\widehat R_{\mathcal E}(\lambda - i0)\big)\widehat f,\widehat g\big)
= (\widehat{\mathcal F}^{(\pm)}(\lambda)\widehat f,\widehat{\mathcal F}^{(\pm)}(\lambda)\widehat g)_{{\bf h}_{\lambda}}.
\nonumber
\end{equation}
We define an operator $\widehat{\mathcal F}^{(\pm)}$ by $(\widehat{\mathcal F}^{(\pm)}\widehat f)(\lambda) = \widehat{\mathcal F}^{(\pm)}(\lambda)\widehat f$, and we also put
\begin{equation}
\widehat{\mathcal F}^{(\pm)} = 0,  \quad {\rm on} \quad {\mathcal H}_p(\widehat H_{\mathcal E});
\nonumber
\end{equation}
this yields the sought spectral representation of  $\widehat H_{\mathcal E}$.

On the other hand, concerning the perturbation of the lattice structure, we take a cut-off function $\chi_0$ whose support contains all the perturbation, and put $\chi_{\infty} = 1 - \chi_0$. In that case the equality
\begin{equation}
\chi_{\infty}\widehat R_{\mathcal E}(\lambda \pm i0) = \widehat R^{(0)}_{\mathcal E}(\lambda \pm i0)\big(\chi_{\infty} + [H^{(0)}_{\mathcal E},\chi_{\infty}]\widehat R_{\mathcal E}(\lambda \pm i0)\big)
\label{ResolventEquation2}
\end{equation}
plays the role of the resolvent equation, and $\widehat{\mathcal F}^{(\pm)}(\lambda)$ is defined by
\begin{equation}
\widehat{\mathcal F}^{(\pm)}(\lambda) = \widehat{\mathcal F}^{(0)}(\lambda)\big(\chi_{\infty} +
[H^{(0)}_{\mathcal E},\chi_{\infty}]\widehat R_{\mathcal E}(\lambda \pm i0)\big).
\label{DefineFjlambdaperturbed}
\end{equation}

Summarizing this discussion, we obtain the following result.

\begin{theorem}
\label{EigenfunctionExpansionWholespace}
(1) The operator $\widehat{\mathcal F}^{(\pm)}$ extends uniquely to a unitary operator from ${\mathcal H}_{ac}(\widehat H_{\mathcal E})$ to ${\mathbb H}$ annihilating the subspace ${\mathcal H}_p(\widehat H_{\mathcal E})$.

\noindent
(2) The operator diagonalizes $\widehat H_{\mathcal E}$, namely
$$
\big(\widehat{\mathcal F}^{(\pm)}\widehat H_{\mathcal E}\widehat f\big)(\lambda) =
\lambda\big(\widehat{\mathcal F}^{(\pm)}\widehat f\big)(\lambda), \quad \forall \widehat f \in D(\widehat H_{\mathcal E}).
$$
(3) The adjoint operator $\widehat{\mathcal F}^{(\pm)}(\lambda)^{\ast} \in {\mathbb B}({\bf h}_{\lambda};\mathcal B^{\ast}(\mathcal E))$ satisfies the eigenequation
$$
(\widehat H_{\mathcal E} - \lambda)\widehat{\mathcal F}^{(\pm)}(\lambda)^{\ast}\phi = 0, \quad \forall \phi \in {\bf h}_{\lambda}.
$$
(4) For any $\widehat f \in \mathcal H_{ac}(\widehat H_{\mathcal E})$, the inversion formula holds,
$$
\widehat f = \int_{\sigma_{ac}(\widehat H_{\mathcal E})}\widehat{\mathcal F}^{(\pm)}(\lambda)^{\ast}\big(\widehat{\mathcal F}^{(\pm)}\widehat f\big)(\lambda)d\lambda.
$$
\end{theorem}
We omit the proof, as it is almost the same as that of Theorem 7.11 in
\cite{AIM16}.

\subsection{Resolvent expansion}

We look at the behavior at infinity of $\widehat R_{\mathcal E}(\lambda \pm i0)\widehat f$ in the sense of $\widehat{\mathcal B}^{\ast}(\mathcal E)$, which is equivalent to observing the singularities of its Fourier transform in the sense of $\mathcal B^{\ast}(\mathcal E)$.
\begin{lemma}
\label{rbfelambdainL2}
For any compact interval $I \subset \big({\rm Int}\,\sigma_e(\widehat H^{(0)}_{\mathcal E})\big) \setminus \sigma^{(0)}_{\mathcal T}$, 
there exists a constant $C > 0$ such that
\begin{equation}
\|\{r^{(0)}_{\bf e}(\lambda)\widehat f_{\bf e}\}_{{\bf e}\in\mathcal E}\|_{\ell^2(\mathcal E)} \leq C\|\widehat f\|_{\ell^2(\mathcal E)}
\nonumber
\end{equation}
holds for all $\lambda \in I$ and ${\bf e} \in \mathcal E$.
\end{lemma}
\begin{proof} 
Since $I$ is in the resolvent set of $- (d/dz)^2_D + V_{\bf e}$, the claim follows. 
\end{proof}

\medskip

For a pair $\widehat f, \widehat g \in \widehat{\mathcal B}^{\ast}(\mathcal E)$, we consider the following equivalence relation
\begin{equation}
\widehat f \simeq \widehat g \Longleftrightarrow
\widehat f - \widehat g \in \widehat{\mathcal B}^{\ast}_0(\mathcal E).
\nonumber
\end{equation}
\begin{lemma}
\label{ResolventEquation2}
For any  $\lambda \in \big({\rm Int}\,\sigma_e(\widehat H^{(0)}_{\mathcal E})\big) \setminus \sigma^{(0)}_{\mathcal T}$ and $\widehat f \in \mathcal B(\mathcal E)$,
we have
\begin{equation}
 \mathcal U_{\mathcal E}\widehat R^{(0)}_{\mathcal E}(\lambda \pm i0)\widehat f
 \simeq \frac{\sin\sqrt{\lambda}}{\lambda} \sum_{j=1}^{s}
\frac{\widehat{\mathcal F}^{(0)}_j(\lambda)\widehat f}{\lambda_j(x) - E(\lambda \pm i0)}.
\nonumber
\end{equation}
\end{lemma}
\begin{proof}
Lemma \ref{rbfelambdainL2} in combination with \eqref{LemmaR0lambdapi0mathcalE} implies
\begin{equation}
\begin{split}
\widehat R^{(0)}_{\mathcal E}(\lambda \pm i0) \hat f &\simeq
\Phi^{(0)}(\lambda)^{\ast}D^{(0)}(\lambda\pm i0)\Phi^{(0)}(\lambda)\hat f \\
& = \frac{\sin\sqrt{\lambda}}{\sqrt{\lambda}}\sum_{j=1}^s
\frac{1}{\lambda_j(x)  - E(\lambda \pm i0)}P_{\mathcal V,j}(x)(\Phi^{(0)}(\lambda)\hat f)(x).
\end{split}
\end{equation}
By virtue of eq.~(4.34) of \cite{AIM16}, we have, for $g \in \mathcal B({\mathbb T}^d)$, the equivalence
$$
\frac{1}{\lambda_j(x) - \mu \mp i0}g(x) \simeq \frac{1}{\lambda_j(x) - \mu \mp i0}g\big|_M,
$$
where $M = \{x \in {\bf T}^d\, ; \, \lambda_j(x) = \mu\}$. This proves the claim.
\end{proof}

Next, we extend this lemma to the perturbed case.

\begin{theorem}
\label{ResolventExpansionPerturbedCase}
For any  $\lambda \in \big({\rm Int}\,\sigma_e(\widehat H_{\mathcal E})\big) \setminus \sigma^{(0)}_{\mathcal T}$ and $\widehat f \in \mathcal B(\mathcal E)$, 
we have
\begin{equation}
 \mathcal U_{\mathcal E}\chi_{\infty}\widehat R_{\mathcal E}(\lambda \pm i0)\widehat f
 \simeq  \frac{\sin\sqrt{\lambda}}{\sqrt{\lambda}}\sum_{j=1}^s
\frac{1}{\lambda_j(x) - E(\lambda \pm i0)} \widehat{\mathcal F}^{(\pm)}_j(\lambda)\widehat f.
\nonumber
\end{equation}
\end{theorem}
\begin{proof}
For the case of lattice structure perturbations, we use the resolvent equation (\ref{ResolventEquation2}).
By Lemma \ref{ResolventEquation2}, the left-hand side is, modulo $\mathcal B^{\ast}_0({\mathbb T}^d)$, equal to
$$
\frac{\sin\sqrt{\lambda}}{\sqrt{\lambda}}\sum_{j=1}^s
\frac{1}{\lambda_j(x)  - E(\lambda \pm i0)}\widehat{\mathcal F}_j^{(0)}(\lambda)
\big(\chi_{\infty} +
[H^{(0)}_{\mathcal E},\chi_{\infty}]\widehat R_{\mathcal E}(\lambda \pm i0)\big)\hat f,
$$
and thus the claim follows from (\ref{DefineFjlambdaperturbed}). For the case of potential perturbations,
we  note that
$$
\mathcal U_{\mathcal E}\chi_{\infty}\widehat R_{\mathcal E}^{(0)}(\lambda \pm i0)\widehat f \simeq
\mathcal U_{\mathcal E}\widehat R_{\mathcal E}^{(0)}(\lambda \pm i0)\widehat f,
$$
since passing to the Fourier series, we see that $(1 - \chi_{\infty})\widehat R_{\mathcal E}^{(0)}(\lambda \pm i0)\widehat f $ is a smooth function on the torus ${\mathbb T}^d$. Then, using (\ref{DefineFpmlambdainmathcalE}) and the resolvent equation (\ref{Resolventequation1}), we obtain the sought result.
\end{proof}

\subsection{Helmholtz equation and S-matrix}

Now one can obtain the asymptotic expansion of solutions to the Helmholtz equation and derive the S-matrix.
\begin{theorem}
\label{HelmhlotzEqWholespace}
(1) For any solution $\widehat u \in \widehat{\mathcal B}^{\ast}(\mathcal E)$ of the equation
$$
(\widehat H_{\mathcal E} - \lambda)\widehat u = 0,
$$
there is an incoming datum and an outgoing datum $\phi^{in}, \phi^{out} \in {\bf h}_{\lambda}$  satisfying
\begin{equation}
\begin{split}
 \mathcal U_{\mathcal E}\chi_{\infty}\widehat u
 \simeq &
- \sum_{j=1}^s\frac{\phi^{in}_j}{\lambda_j(x) - E(\lambda - i0)} +  \sum_{j=1}^s\frac{\phi^{out}_j}{\lambda_j(x) - E(\lambda + i0)}.
\end{split}
\label{TheoremBastcharacterize1}
\end{equation}
(2)
For any incoming datum $\phi ^{in} = (\phi^{in}_1,\dots,\phi^{in}_s) \in {\bf h}_{\lambda}$, there exist a unique solution $\widehat u \in \widehat{\mathcal B}^{\ast}(\mathcal E)$ of the equation
$$
(\widehat H_{\mathcal E} - \lambda)\widehat u = 0
$$
and an outgoing datum $\phi^{out} = (\phi^{out}_1,\dots,\phi^{ou}_s) \in {\bf h}_{\lambda}$  satisfying the relation (\ref{TheoremBastcharacterize1}).
The operator $S(\lambda)$ defined by
\begin{equation}
S(\lambda) : \phi^{in}  \to \phi^{out}
\nonumber
\end{equation}
 is unitary on ${\bf h}_{\lambda}$.
\end{theorem}
\begin{proof}
Let $\hat u \in \widehat{\mathcal B}^{\ast}(\mathcal E)$ be a solution to $(\widehat H_{\mathcal E} - \lambda)\hat u = 0$ and put ${\widehat u}\big|_{\mathcal V} = \widehat w$. Then, $\widehat w \in \widehat{\mathcal B}^{\ast}(\mathcal V)$ and satisfies $(- \widehat\Delta_{\mathcal V,\lambda} + \cos\sqrt{\lambda})\widehat w = 0$. By virtue of Theorem \ref{TheoremHelmhotzeqmathcalV}(1), this $\widehat w$ admits an asymptotic expansion\footnote{Note that we have to replace $- \widehat\Delta_{\Gamma}$ by $- \widehat\Delta_{\mathcal V,\lambda}$ and  the energy parameter $\lambda$  by $E(\lambda)$ in (\ref{Expansionofuforvertexcase}).} (\ref{Expansionofuforvertexcase}). As $\widehat u = \widehat T_{\mathcal V,\lambda}^{\ast}\widehat w$, the first claim follows.

The existence part of (2) can be proven by the same argument as above, reducing it to the case of the vertex operator. To prove the uniqueness, we take $\phi^{in} = 0$, and consider the solution $\hat u \in \widehat{\mathcal B}^{\ast}(\mathcal E)$ of the equation $(\widehat H_{\mathcal E} - \lambda)\hat u = 0$ such that
\begin{equation}
\begin{split}
 \mathcal U_{\mathcal E}\chi_{\infty}\hat u
 \simeq &  \sum_{j=1}^s\frac{\phi^{out}_j}{\lambda_j(x) - E(\lambda + i0)}.
\end{split}
\label{TheoremBastcharacterize2}
\end{equation}
Then $\widehat u$ satisfies the outgoing radiation condition, and by Lemma \ref{LemmaRadCondUnique}, such a solution vanishes identically.
\end{proof}

As this argument shows, the S-matrix for $\widehat H_{\mathcal E}$ at the energy $\lambda$ coincides with the S-matrix for $- \widehat\Delta_{\mathcal V,\lambda}$ at the energy ${\frac{\sqrt{\lambda}}{\sin\sqrt{\lambda}}}E(\lambda)$, and hence the unitarity follows. Stated more explicitly, we conclude:
\begin{cor}
\label{CorSmatrixedgeSmatrixvertex}
The S-matrix for $\widehat H_{\mathcal E}$ at the energy $\lambda$ coincides with the S-matrix for $- \widehat\Delta_{\mathcal V,\lambda}$ at the energy $\displaystyle{{\frac{\sqrt{\lambda}}{\sin\sqrt{\lambda}}}E(\lambda)=- {\sqrt{\lambda}\cot{\sqrt{\lambda}}} - \kappa_{\mathcal V}}$.
\end{cor}
\begin{remark}
By checking the above proof, one can see that all the arguments in this section remain valid in the situation when $C_v/d_v$ is a fixed constant except for a finite number of vertices $v \in \mathcal V$. Moreover, one can deal in the same way with the case where the unperturbed operator $h_e^{(0)}$ has the same potential $V_0(z)$ at all the edges, that is, $h^{(0)}_e = - (d^2/dz^2)_D + V_0(z)$, $\forall e \in \mathcal E$.
\end{remark}

 \section{Inverse scattering for quantum graph}
\label{InverseQuantumGraph}

\begin{theorem}
\label{QuantumgraphSmatrixandDnMapTheorem}
For the Schr{\"o}dinger operator $\widehat H_{\mathcal E}$ on a quantum graph of the considered class, the S-matrix $S(\lambda)$ for the scattering problem and the D-N map $\Lambda_{\mathcal E}(\lambda)$ for the interior boundary value problem determine each other.
\end{theorem}
\begin{proof}
By Corollary \ref{CorSmatrixedgeSmatrixvertex}, knowing the S-matrix $S(\lambda)$ for $\widehat H_{\mathcal E}$ is equivalent to knowing the S-matrix for $- \widehat\Delta_{\mathcal V,\lambda}$ at the energy $\frac{\sqrt{\lambda}}{\sin\sqrt{\lambda}}E(\lambda)$. By Theorem \ref{ThSmatrixDNmapEquiv}, this is equivalent to knowing the D-N map for $- \widehat{\Delta}_{\mathcal V,\lambda}$ at the energy ${\frac{\sqrt{\lambda}}{\sin\sqrt{\lambda}}}E(\lambda)$. Finally by Lemma \ref{DNmathcalE=DNmathcalV}, this is equivalent to knowing the D-N map for $\widehat H_{\mathcal E}$ at the energy $\lambda$.
\end{proof}

We have now arrived at our next main theorem.
\begin{theorem}
\label{qunatumgraphstrutureTheoremScattering}
Let   $\Gamma = \{\mathcal V, \mathcal E\}$ and $\Gamma' = \{\mathcal V',\mathcal E'\}$ be two infinite quantum graphs 
 as in \S \ref{InverseScatteringDiscrete} satisfying (\ref{S3EquilateralCond}), (\ref{S3Equivertexpotentialcond}), and (D-1)--(D-4),
whose  perturbed finite subgraphs satisfy (C-1), (C-2).  Assume further that  $\ell_{\mathcal E} = \ell_{\mathcal E'}, V_{\mathcal E}(z) = V_{\mathcal E'}(z)$,  $k_{\mathcal V}= k_{\mathcal V'}$. 
Suppose that the S-matrices for the Schr{\"o}dinger operator for the two quantum graphs coincide for all energies.
Then there is a bijection $\Phi : \Gamma \to \Gamma'$ preserving the graph structure, and $d_v = d_{v'}$, $C_v = C'_{v'}$ hold for all $v \in \mathcal V$ and $v' = \Phi(v)$.
\end{theorem}

\begin{proof}
This is a direct consequence of Theorems \ref{qunatumgraphstrutureTheoremBVP}  and \ref{QuantumgraphSmatrixandDnMapTheorem}.
\end{proof}

 \end{document}